\documentclass[10pt,journal,compsoc]{IEEEtran}
\ifCLASSOPTIONcompsoc
  \usepackage[nocompress]{cite}
\else
  \usepackage{cite}
\fi
\ifCLASSINFOpdf
\else
\fi
\usepackage{amsmath}
\usepackage{color}
\usepackage{graphicx}
\usepackage{threeparttable}
\usepackage[linesnumbered, ruled, vlined]{algorithm2e}
\usepackage{amsmath}
\usepackage{amssymb}
\usepackage{amsthm}
\usepackage{mathtools}

\newtheorem{theorem}{\textbf{Theorem}}
\newtheorem{definition}{\textbf{Definition}}
\newtheorem{feature}{Feature}
\newtheorem{proposition}{\textbf{Proposition}}
\newtheorem{lemma}{Lemma}
\newtheorem{property}{\textbf{{\em Property}}}

\usepackage[noend]{algpseudocode}
\usepackage{indentfirst}

\newcommand{\tth}{t^{\textrm{th}}}

\makeatletter
\newcommand{\rmnum}[1]{\romannumeral #1}
\newcommand{\Rmnum}[1]{\expandafter\@slowromancap\romannumeral #1@}
\makeatother

\usepackage{array}
\newcolumntype{C}[1]{>{\centering\arraybackslash}m{#1}}

\begin{document}
%
\title{Algorithms for Scheduling Malleable Tasks}
%
%
%
%

\author{Xiaohu~Wu,
        and~Patrick~Loiseau
\IEEEcompsocitemizethanks{\IEEEcompsocthanksitem Xiaohu Wu is with Fondazione Bruno Kessler, Trento,
Italy.\protect\\
E-mail: xiaohuwu@fbk.eu
\IEEEcompsocthanksitem Patrick Loiseau is with Univ. Grenoble Alpes, LIG, France and MPI-SWS, Germany. E-mail: patrick.loiseau@univ-grenoble-alpes.fr


}
\thanks{Manuscript received April 19, 2005; revised August 26, 2015.}}

%
%

\markboth{Journal of \LaTeX\ Class Files,~Vol.~14, No.~8, August~2015}%
{Shell \MakeLowercase{\textit{et al.}}: Bare Demo of IEEEtran.cls for Computer Society Journals}
%



\IEEEtitleabstractindextext{%
\begin{abstract}
Due to the ubiquity of batch data processing in cloud computing, the related problems of scheduling malleable batch tasks have received significant attention recently. In this paper, we consider a fundamental model where a set of $n$ tasks is to be processed on $C$ identical machines and each task is specified by a value, a workload, a deadline and a parallelism bound. Within the parallelism bound, the number of machines assigned to a task can vary over time without affecting its workload. For this model, we first give two core results: the definition of an optimal state under which multiple machines could be utilized by a set of tasks with hard deadlines, and, an algorithm achieving such a state.
The optimal utilization state plays a key role in the design and analysis of scheduling algorithms
(\rmnum{1}) when several typical objectives are considered, such as social welfare maximization, machine minimization, and minimizing the maximum weighted completion time, and, (\rmnum{2}) when the algorithmic design techniques such as greedy and dynamic programming are applied to the social welfare maximization problem. As a result, we give four new or improved algorithms for the above problems.
\end{abstract}

}

\maketitle

\IEEEdisplaynontitleabstractindextext

%
\IEEEpeerreviewmaketitle

\section{Introduction}
\label{intro}

Cloud computing has become the norm for a wide range of applications and batch processing constitutes the most significant computing paradigm \cite{Hu}. Applications such as web search index update, monte carlo simulations and big-data analytics require executing a new type of parallel tasks on clusters, termed {\em malleable tasks}. Two basic features of malleable tasks are about {\em workload} and {\em parallelism bound}. There are multiple machines, and, throughout the execution, the number of machines assigned to a task can vary over time within the parallelism bound but its workload is not affected by the number of used machines \cite{Jain11a,Jain}. Beyond understanding how to schedule the fundamental batch task model, many efforts are also devoted to its online version \cite{Lucier,Azar,Jain14} and its extension in which each task contains several subtasks with precedence constraints \cite{Bodik,Nagarajan}. In practice, for better efficiency, companies such as IBM have integrated these smarter scheduling algorithms for various time metrics  \cite{Nagarajan} (than the popular dominant resource fairness strategy) into their batch processing platforms \cite{Wolf}.

In scheduling theory, the above malleable task model can be viewed as an extension of the classic model of scheduling preemptive tasks on a single or multiple machines where the parallelism bound is one \cite{Lawler91,Karger}.  
When each task has to be completed by some deadline, the results from the special single machine case have already implied that the state of optimally utilizing machines plays a key role in the design and analysis of scheduling algorithms under several objectives \cite{Karger}. In particular, the famous EDF (Earliest Deadline First) rule can achieve an optimal schedule for the single machine case. It is initially designed so as to find an exact algorithm for scheduling batch tasks to minimize the maximum task lateness (i.e., task's completion time minus due date) \cite{Jac55}. So far, numerous applications of this rule have been found, e.g., (\rmnum{1}) to design exact algorithms for the extended model with release times \cite{Hor74} and for scheduling tasks with deadlines (and release times) to minimize the total weighted number of tardy tasks \cite{Lawler69}, and (\rmnum{2}) as a significant principle in the analysis of scheduling feasibility for real-time systems \cite{Stankovic98}.

Similarly, we are convinced that, as far as malleable tasks are concerned, achieving such an optimal resource utilization state is also very important for designing and analyzing scheduling algorithms (\rmnum{1}) under various objectives, or (\rmnum{2}) when different algorithmic design techniques such as greedy and dynamic programming are applied. The intuition for this is that, if the utilization state was not optimal in an algorithm, its performance could be improved by utilizing the machines optimally to allow more tasks to be completed. All these considerations motivate us to develop an theoretical framework proposed in this paper.



Before this paper, a greedy algorithm was proposed in \cite{Jain} that achieves a performance guarantee $\frac{C-k}{C}\cdot\frac{s-1}{s}$; here, $C$ is the number of machines, $k$ is the maximum parallelism bound of all tasks, $s$ is the minimum slackness of all tasks where each task's slackness is defined to be the ratio of its deadline to its minimum execution time, which is the time when a task is always allocated the maximum number of machines during the execution. $k$ is a system parameter and is assumed to be finite \cite{White}. Intuitively, $s$ characterizes the resource allocation urgency (e.g., $s=1$ means that the maximum amount of machines have to be allocated to a task at every time slot to meet its deadline).

\subsection{Our Results}
\label{sec.contributions}


\vspace{0.22em}\noindent\textbf{Core result} (Section~\ref{sec.optimal}). The core result of this paper is to identify a sufficient and necessary condition under which a set of independent malleable tasks could be all completed by their deadlines on $C$ machines, also referred to as boundary condition in this paper.



In particular, by understanding the basic constraints of malleable tasks, we first identify and formally define a state in which $C$ machines can be said to be optimally utilized by a set of tasks with deadlines in terms of resource utilization. Then, we propose an optimal scheduling algorithm LDF($\mathcal{S}$) (Latest Deadline First) that achieves such an optimal state. The LDF($\mathcal{S}$) algorithm has a polynomial time complexity of $\mathcal{O}(n^{2})$ and is different from the EDF algorithm that gives an optimal schedule in the single-machine case. Here, the maximum deadline of tasks is assumed to be finitely bounded by a constant.

\vspace{0.22em}\noindent\textbf{Applications} (Sections~\ref{more-application} and~\ref{more-app-2}). The above core results have several applications to propose new or improved algorithmic design and analysis for scheduling malleable tasks under different objectives. The scheduling objectives considered in this paper include:
\begin{enumerate}
\item [(a)] \textit{social welfare maximization:} maximize the sum of values of tasks completed by their deadlines;

\item [(b)] \textit{machine minimization:} minimize the number of machines needed to produce a feasible schedule for a set of tasks such that each task is completed by their deadline;

\item [(c)] \textit{maximum weighted completion time minimization:} minimize the maximum weighted completion time of tasks.
\end{enumerate}
Here, the first and second objectives above have been considered in \cite{Jain11a,Jain,Nagarajan}. The second objective that concerns the optimal utilization of machines has been considered for other types of tasks \cite{Chuzhoy04a} but we are the first to consider it for malleable tasks. After applying the core results above, we obtain the following algorithmic results:
\begin{itemize}
\item [(\rmnum{1})] an improved greedy algorithm GreedyRLM with a performance guarantee $\frac{s-1}{s}$ for social welfare maximization with a time complexity of $\mathcal{O}(n^{2})$;
\item [(\rmnum{2})] the first exact dynamic programming algorithm for social welfare maximization with a pseudo-polyno- mial time complexity of $\mathcal{O}(\max\{nd^{L}C^{L},$ $n^{2}\})$, where $L$ is the number of deadlines, $D$ and $d$ are the maximum workload and deadline of tasks;
\item [(\rmnum{3})] the first exact algorithm for machine minimization with a time complexity of $\mathcal{O}(n^{2}, L n \log{n})$;
\item [(\rmnum{4})] a polynomial time (1+$\epsilon$)-approximation algorithm for maximum weighted completion time minimization.
\end{itemize}

In the greedy algorithm of \cite{Jain}, the tasks are considered in the non-decreasing order of their marginal values of tasks (i.e., the ratio of a task's value to its size), and only if a task could be fully completed by its deadline according to the currently remaining resource, it will be accepted and allocated possibly different number of machines over time according to an allocation algorithm; otherwise, it will be rejected. In this paper, we also show that
\begin{itemize}
\item for social welfare maximization, $\frac{s-1}{s}$ is the best possible performance guarantee that a class of greedy algorithms could achieve where they consider tasks in the non-increasing order of their marginal values.
\item as a result, the proposed greedy algorithm of this paper is the best possible among this kind of greedy algorithms.
\end{itemize}

The second algorithm for social welfare maximization can work efficiently when $L$ is small since its time complexity is exponential in $L$. However, this may be reasonable in a machine scheduling context. In scenarios like \cite{Bodik}, tasks are often scheduled periodically, e.g., on an hourly or daily basis, and many tasks have a relatively soft deadline (e.g., finishing after four hours instead of three will not trigger a financial penalty). Then, the scheduler can negotiate with the tasks and select an appropriate set of deadlines $\{\tau_{1}, \tau_{2}, \cdots, \tau_{L}\}$, thereafter rounding the deadline of a task down to the closest $\tau_{i}$ ($1\leq i\leq L$). By reducing $L$, this could permit to use the dynamic programming (DP) algorithm rather than GreedyRLM in the case where the slackness $s$ is close to 1. With $s$ close to 1, the approximation ratio of GreedyRLM approaches 0 and possibly little social welfare is obtained by adopting GreedyRLM while the DP algorithm can still obtain the almost optimal social welfare.

\vspace{0.22em}\noindent\textbf{Technical Difference.} The second algorithm can be viewed as an extension of the pseudo-polynomial time exact algorithm in the single machine case \cite{Lawler91} that is also designed via the generic dynamic programming procedure. However, before our work, how to enable this extension to malleable tasks was not clear as indicated in \cite{Jain11a,Jain}. This is mainly due to the lack of a notion of the optimal state of machines being utilized by malleable tasks with deadlines and the lack of an algorithm that achieves such a state. In contrast, the optimal state in the single machine case can be defined much more easily and achieved by the EDF algorithm. The core results of this paper are the enabler of a DP algorithm.

The way of applying the core results to a greedy algorithm is less obvious since in the single machine case there is no corresponding algorithm to hint its role in the algorithmic design. For the above class of greedy algorithms, we manage to give a new algorithm analysis, figuring out what resource allocation features of tasks can benefit and determine the algorithm's performance. This analysis is an extended analysis of the greedy algorithm for the standard knapsack problem \cite{Brassard} and it does not rely on the dual-fitting technique, on which the algorithm in \cite{Jain} is built.
Here, the problem could be viewed as an extension of the knapsack problem where each item has two additional constraints in a two-dimensional space: a (time) window in which an item could be placed and a maximum width of the space that it could utilize at every moment. Two of the most important algorithms there are either based on the DP technique or of greedy type, that also considers items by their marginal values \cite{Brassard}; we give in this paper their counterparts in the scenario of malleable tasks.

In the construction of the greedy and optimal scheduling algorithms, we are inspired by the algorithm in \cite{Jain}. After our definition of the optimal state and a new analysis of the above class of greedy algorithms, we found that the algorithm in \cite{Jain} could achieve an optimal resource utilization state from the maximum deadline of tasks $d$ to some earlier time slot $t$. However, this is achieved by guaranteeing the existence of a time slot $t^{\prime}$ earlier than $t$ such that the number of available machines at $t^{\prime}$ is $\geq k$, which leads a suboptimal utilization of resources. In our algorithm, we only require $t^{\prime}$ to be such that the number of available machines at $t^{\prime}$ is $\geq 1$, which leads to an optimal resource utilization. More details could be found in the remarks of Section~\ref{sec.scheduling}.

The above third and fourth algorithms are obtained by respectively applying the above core result to a binary search procedure, and the related results in \cite{Nagarajan}.


\subsection{Related works}
\label{sec.related-work}

Now, we introduce the related works.
The linear programming approaches to designing and analyzing algorithms for the task model of this paper \cite{Jain11a,Jain} and its variants \cite{Lucier,Bodik,Azar} have been well studied\footnote{We refer readers to \cite{Karger,Williamson} for more details on the general techniques to design scheduling algorithms.}. All these works consider the same objective of maximizing the social welfare. 
In \cite{Jain11a}, Jain {\em et al.} proposed an algorithm with an approximation ratio of $(1+\frac{C}{C-k})(1+\epsilon)$ via {\em deterministic rounding of linear programming}. Subsequently, Jain {\em et al.} \cite{Jain} proposed a greedy algorithm GreedyRTL and used the {\em dual-fitting technique} to derive an approximation ratio $\frac{C-k}{C}\cdot\frac{s-1}{s}$.
In \cite{Bodik}, Bodik {\em et al.} considered an extension of our task model, i.e., DAG-structured malleable tasks, and, based on {\em randomized rounding of linear programming}, they proposed an algorithm with an expected approximation ratio of $\alpha(\lambda)$ for every $\lambda > 0$, where $\alpha(\lambda)=\frac{1}{\lambda}\cdot e^{-\frac{1}{\lambda}}\cdot \left[ 1-e^{-\frac{(1-1/\lambda)C-k}{2\omega \kappa}\cdot \ln{\lambda\cdot (1-\frac{\kappa}{C})}} \right]$. The online version of our task model is considered in \cite{Lucier,Azar}; again based on the {\em dual-fitting technique}, two weighted greedy algorithms are proposed respectively for non-committed and committed scheduling and achieve the competitive ratios of $cr_{\mathcal{A}}=2+\mathcal{O}(\frac{1}{(\sqrt[3]{s}-1)^{2}})$ where $s>1$ \cite{Jain} and $\frac{cr_{\mathcal{A}}\left( s\cdot\omega(1-\omega) \right)}{\omega(1-\omega)}$ where $\omega\in(0, 1)$ and $s>\frac{1}{\omega(1-\omega)}$.

In addition, Nagarajan {\em et al.} \cite{Nagarajan} considered DAG-structured malleable tasks and propose two algorithms with approximation ratios of 6 and 2 respectively for the objectives of minimizing the total weighted completion time and the maximum weighted lateness of tasks. Nagarajan {\em et al.} showed that {\em optimally scheduling deadline-sensitive malleable tasks in terms of resource utilization is a key to the solutions to scheduling for their objectives.} In particular, seeking a schedule for DAG tasks can be transformed into seeking a schedule for tasks with simpler chain-precedence constraints; then whenever there is a feasible schedule to complete a set of tasks by their deadlines, Nagarajan {\em et al.} proposed a non-optimal algorithm where each task is completed by at most 2 times its deadline and give two procedures to obtain near-optimal completion times of tasks in terms of the above two objectives.

%

Technically, the works \cite{Jain11a,Jain,Lucier,Bodik,Azar} formulate their problem as an Integer Program (IP) and relax the IP to a relaxed linear program (LP). The techniques in \cite{Jain11a,Bodik} require to solve the LP to obtain a fractional optimal solution and then manage to round the fractional solution to an integer solution of the IP that corresponds to an approximate solution to their original problem. In \cite{Jain,Lucier,Azar}, the dual fitting technique first finds the dual of the LP and then construct a feasible algorithmic solution $X$ to the dual in some greedy way. This solution corresponds to a feasible solution $Y$ to their original problems, and, due to the weak duality, the value of the dual under the solution $X$ (expressed in the form of the value under $Y$ multiplied by a parameter $\alpha\geq 1$) will be an upper bound of the optimal value of the IP, i.e., the optimal value that can be achieved in the original problem. Therefore, the approximation ratio of the algorithm involved in the dual becomes clearly  $1/\alpha$. Here, the approximation ratio is a lower bound of the ratio of the actual value obtained by the algorithm to the optimal value.


A part of results of this paper appeared at the Allerton conference in the year 2015 \cite{Wu15a,Wu15b}. Following \cite{Wu15a,Wu15b}, a recent work also gave a similar (sufficient and necessary) feasibility condition to determine whether a set of malleable tasks could be completed by their deadlines and showed that such a condition is central to the application of the LP technique to the three problems of this paper: greedy and exact algorithms for social welfare maximization and an exact algorithm for machine minimization. Guo \& Shen first used the LP technique to give a new proof of this feasibility condition in the core result. Based on this condition, the authors gave a new formulation of the original problems as IP programs, different from the ones in \cite{Jain11a,Jain}. This new formulation enables from a different perspective proposing almost the same algorithmic results as this paper, e.g., for the machine minimization problem an exact algorithm with a time complexity $\mathcal{O}((n+d)^{3.5}L_{s}(\log{n}+\log{k}))$, and for the social welfare maximization problem an exact algorithm with a complexity $\mathcal{O}(n\cdot (C\cdot d)^{d})$, where $L_{s}$ is the length of the LP's input. In addition, we have shown that the best performance guarantee is $\frac{s-1}{s}$ when a greedy algorithm considers tasks in the non-increasing order of their marginal values. Guo \& Shen also considered another standard to determine the order of tasks, and proposed a greedy algorithm with a performance guarantee $\frac{C-k}{C}$ and a complexity $\mathcal{O}(n^{2}+nd)$.

\section{Model and Problem Description}
\label{sec.model}

\begin{table}
\centering
\begin{threeparttable}[!ht]

\caption{Main Notation}
\begin{tabular}{|C{1.5cm}|C{6.2cm}|}
\hline

   Notation & Explanation\\ \hline

$C$ & the total number of machines \\ \hline

$\mathcal{T}$ & a set of tasks to be scheduled on $C$ machines \\ \hline

$T_{i}$ & a task in $\mathcal{T}$ \\ \hline

$D_{i}, d_{i}, v_{i}$ & the workload, deadline, and value of a task $T_{i}$ \\ \hline

$k_i$ & the parallelism bound of $T_{i}$, i.e., the maximum number of machines that can be allocated to and utilized by $T_{i}$ simultaneously \\ \hline

$y_{i}(t)$ & the number of machines allocated to $T_{i}$ at a time slot $t$ where $y_{i}(t)\in \{0,1,\cdots,k_{i}\}$ and set all $y_{i}(t)$ to $0$ initially \\ \hline

$W(t)$ & the total number of machines that are allocated out to the tasks at $t$, i.e., $W(t)=\sum_{T_{i}\in\mathcal{T}}{y_{i}(t)}$ \\ \hline

$\overline{W}(t)$ &  the total number of machines idle at $t$, i,e., $\overline{W}(t) = C - W(t)$ \\ \hline

$len_{i}$ & the minimum execution time of $T_{i}$ where $T_{i}$ is allocated $k_{i}$ machines in the entire execution process, i.e., $len_{i}=\lceil \frac{D_{i}}{k_{i}} \rceil$\\ \hline

$s_{i}$ & the slackness of a task, i.e., $\frac{d_{i}}{len_{i}}$, measuring the urgency of machine allocation to complete $T_{i}$ by the deadline \\ \hline

$s$  &  the minimum slackness of all tasks of $\mathcal{T}$, i.e., $\min_{T_{i}\in\mathcal{T}}{s_{i}}$ \\ \hline

$d$, $D$ & the maximum deadline and workload of all tasks of $\mathcal{T}$, i.e., $d=\max_{T_{i}\in\mathcal{T}}{d_{i}}$ and $D=\max_{T_{i}\in\mathcal{T}}{D_{i}}$ \\ \hline

$v_{i}^{\prime}$ & the marginal value of $T_{i}$, i.e., $v_{i}^{\prime}=\frac{v_{i}}{D_{i}}$ \\ \hline

$\{\tau_{1}, \cdots, \tau_{L}\}$ & the set of the deadlines $d_{i}$ of all tasks $T_{i}$ of $\mathcal{T}$, where $0=\tau_{0} < \tau_{1}< \cdots < \tau_{L}=d$ \\ \hline

$\mathcal{D}_{i}$ & all the tasks $\{T_{i,1}, T_{i,2}, \cdots, T_{i,n_{i}} \}$ of $\mathcal{T}$ that have a deadline $\tau_{i}$, $1\leq i\leq L$ \\ \hline

\end{tabular}
\label{table-1}
 \end{threeparttable}
\end{table}

There are $C$ identical machines and a set of $n$ tasks $\mathcal{T} = \{T_{1}, T_{2}, \cdots, T_{n}\}$. The task $T_{i}$ is specified by several characteristics: (1) {\em value} $v_{i}$, (2) {\em demand} (or {\em workload}) $D_{i}$, (3) {\em deadline} $d_{i}$, and (4) {\em parallelism bound} $k_{i}$. Time is discrete and the time horizon is divided into $d$ time slots: $\{1, 2, \cdots, d\}$, where $d=\max_{T_{i}\in\mathcal{T}}{d_{i}}$ and the length of each slot may be a fixed number of minutes. A task $T_{i}$ can only utilize the machines located in time slot interval $[1, d_{i}]$. The parallelism bound $k_{i}$ limits that, at any time slot $t$, $T_{i}$ can be executed on at most $k_{i}$ machines simultaneously. Let $k=\max_{T_{i}\in\mathcal{T}}{k_{i}}$ be the maximum parallelism bound; here, $k_{i}$ is a system parameter and $k$ is therefore assumed to be finite \cite{White}. An {\em allocation} of machines to a task $T_{i}$ is a function $y_{i}: [1, d_{i}]\rightarrow \{0, 1, 2, \cdots, k_{i}\}$, where $y_{i}(t)$ is the number of machines allocated to task $T_i$ at a time slot $t\in [1, d_{i}]$. In this model, $D_{i}, d_{i}\in\mathcal{Z}^{+}$ for all $T_{i}\in\mathcal{T}$.

For the system of $C$ machines, denote by $W(t)=\sum_{T_{i}\in\mathcal{T}}{y_{i}(t)}$ the system's workload at time slot $t$; and by $\overline{W}(t)=C-W(t)$ its complementary, i.e., the amount of available machines at time $t$. We say that time $t$ is {\em fully utilized} if $\overline{W}(t)=0$, and is {\em not fully utilized} if $\overline{W}(t)>0$. In addition, we assume that the maximum deadline of tasks is bounded. Given the model above, the following three scheduling objectives are considered separately in this paper:
\begin{itemize}
  \setlength\itemsep{0.17em}
\item \textit{The first objective} is social welfare maximization and it aims to choose an a subset $\mathcal{S}\subseteq\mathcal{T}$ and produce a feasible schedule for $\mathcal{S}$ so as to maximize the social welfare $\sum_{T_{i}\in \mathcal{S}}{v_{i}}$ (i.e., the sum of values of tasks completed by deadlines); here, the value $v_i$ of a task $T_{i}$ is gained if and only if it is {\em fully allocated} by the deadline, i.e., $\sum_{t\leq d_{i}}{y_{i}(t)} \geq D_{i}$, and partial execution of a task yields no value.
\item \textit{The second objective} is machine minimization, i.e., seeking the minimum number of machines needed to produce {\em a feasible schedule} of $\mathcal{T}$ on $C$ machines such that the task's parallelism bound and deadline constraints are not violated.
\item \textit{The third objective} is to minimize the maximum weighted lateness of tasks, i.e., $\min_{T_{i}\in\mathcal{T}}\{v_{i}\cdot(t_{i}-d_{i})\}$, where $t_{i}$ is the completion time of a task $T_{i}$.
\end{itemize}
Furthermore, we denote by $[l]$ and $[l]^{+}$ the sets $\{0, 1, \cdots, l\}$ and $\{1, 2, \cdots, l\}$ for a positive integer $l$. Let $len_{i} = \left\lceil D_i/k_i \right\rceil$ denote the {\em minimum execution time} of $T_{i}$. Define by $s_{i}= \frac{d_{i}}{len_{i}}$ the slackness of $T_{i}$, measuring the urgency of machine allocation (e.g., $s_{i}=1$ may mean that $T_{i}$ should be allocated the maximum amount of machines $k_{i}$ at every $t\in [1, d_{i}]$) and let $s=\min_{T_{i}\in\mathcal{T}}{s_{i}}$ be the slackness of the least flexible task ($s\ge 1$). Denote by $v_{i}^{\prime}=\frac{v_{i}}{D_{i}}$ the {\em marginal value} of task $T_{i}$, i.e., the value obtained by the system {\em per unit of demand}. We assume that the demand of each task is an integer. Let $D=\max_{T_{i}\in\mathcal{T}}\{D_{i}\}$ be the demand of the largest task. Given a set of tasks $\mathcal{T}$, the deadlines $d_{i}$ of all tasks $T_{i}\in\mathcal{T}$ constitute a finite set $\{\tau_{1}, \tau_{2}, \cdots, \tau_{L}\}$, where $L\leq n$, $\tau_{1},\cdots,\tau_{L}\in\mathcal{Z}^{+}$, and $0=\tau_{0} < \cdots < \tau_{L}=d$. Let $\mathcal{D}_{i}=\{T_{i,1}, T_{i,2}, \cdots, T_{i,n_{i}} \}$ denote the set of tasks with deadline $\tau_{i}$, where $\sum_{i=1}^{L}{n_{i}}=n$ ($i\in[L]^{+}$).

The notation of this section is used in the entire paper and summarized in Table~\ref{table-1}. Throughout this paper, we use $i$, $j$, $m$, $l$, or $m^{\prime}$ as subscripts to index the element of different sets such as tasks and use $t$ or $\overline{t}$ to index a time slot.

\section{Optimal Schedule}
\label{sec.optimal}

In this section, we identify a state under which $C$ machines can be said to be optimally utilized by a set of tasks. We then propose a scheduling algorithm that achieves such an optimal state. Besides Table~\ref{table-1}, the additional notation to be used in this section is summarized in Table~\ref{table-2}. 


\subsection{Optimal Resource Utilization State}
\label{sec.optimal-utilization}

In this paper, all tasks are denoted by a set $\mathcal{T}$, and we denote by $\mathcal{S}\subseteq\mathcal{T}$ an arbitrary subset of $\mathcal{T}$; all tasks of $\mathcal{T}$ with a deadline $\tau_{l}$ are denoted by $\mathcal{D}_{l}$ and we denote by $\mathcal{S}_{l}=\mathcal{S}\cap\mathcal{D}_{l}$ all tasks of $\mathcal{S}$ with a deadline $\tau_{l}$ ($l\in[L]^{+}$). In this subsection, we define the maximum amount of workload of $\mathcal{S}$ that could be processed in a fixed time interval $[\tau_{m}+1, \tau_{L}]$ on $C$ machines for all $m\in[L-1]$, where $\tau_{L}=d$, i.e., the maximum deadline of tasks.



\begin{figure}[!ht]
  \centering

  \includegraphics[width=3.2in]{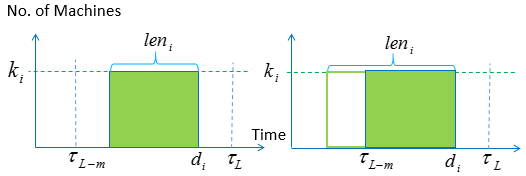}

  \caption{The green areas denote the maximum demand of $T_{i}$ that need or could be processed in $[\tau_{L-m}+1, \tau_{L}]$.}\label{Fig.1}
\end{figure}

We first define {\em the maximum amount of resource, denoted by $\lambda_{m}(\mathcal{S})$, that could be utilized by $\mathcal{S}$ in $[\tau_{L-m}+1, \tau_{L}]$ in an idealized case} where there is an indefinite number of machines, i.e., $C=\infty$, for all $m\in[L]^{+}$. To define this, we clarify the maximum amount of resource that an individual task $T_{i}$ can utilize in $[\tau_{L-m}+1, \tau_{L}]$. The basic constraints of malleable tasks with deadlines imply that:
\begin{itemize}
\item the deadline of $T_{i}$ limits that $T_{i}$ can only utilize the machines in $[1, d_{i}]$, and

\item the parallelism bound limits that $T_{i}$ can only utilize at most $k_{i}$ machines simultaneously at every time slot.
\end{itemize}
The tasks with $d_{i}\leq \tau_{L-m}$ cannot be executed in the interval $[\tau_{L-m}+1, \tau_{L}]$. Let us consider a task $T_{i}$ with $d_{i}\in [\tau_{L-m}+1, \tau_{L}]$. The number of time slots available in $[\tau_{L-m}+1, d_{i}]$ is $d_{i}-\tau_{L-m}$ in the discrete case, and, also recall that $len_{i}$ the (minimum) execution time of $T_{i}$ when it always utilizes the maximum number $k_{i}$ of machines throughout the execution. In the illustrative Fig.~\ref{Fig.1}, the green area in the left (resp. right) subfigure denotes the maximum demand of a task, i.e., $D_{i}$ (resp. $k_{i}\cdot(d_{i}-\tau_{L-m})$), that could or need be processed in $[\tau_{L-m}+1, \tau_{L}]$ in the case where the minimum execution time is such that $len_{i}\leq  d_{i}-\tau_{L-m}$ (resp. $len_{i} > d_{i}-\tau_{L-m}$).

As a consequence of the observation above, $\lambda_{m}(\mathcal{S})$ equals the sum of the maximum workload of every task in $\mathcal{S}$ that could executed in $[\tau_{L-m}+1, \tau_{L}]$ and is defined as follows.
\begin{definition}\label{Def-1}
Initially, set $\lambda_{m}(\mathcal{S})$ to zero for all $m\in [L]$. In the case where $C=\infty$ (i.e., the capacity constraint is ignored), for all $m\in[L]^{+}$, $\lambda_{m}(\mathcal{S})$ is defined as follows:
\begin{center}
$\lambda_{m}(\mathcal{S})\leftarrow \lambda_{m}(\mathcal{S})+\beta_{i}$, for every task $T_{i}\in\mathcal{S}$,
\end{center}
where $\beta_{i}$ is such that
\begin{itemize}
  \setlength\itemsep{0.2em}

\item if $T_{i}\in\mathcal{S}_{1}\cup\cdots\cup \mathcal{S}_{L-m}$ where $d_{i}\leq \tau_{L-m}$, $\beta_{i}\leftarrow 0$;

\item if $T_{i}\in\mathcal{S}_{L-m+1}\cup\cdots\cup \mathcal{S}_{L}$ where $d_{i}\geq \tau_{L-m}+1$, as illustrated in Figure~\ref{Fig.1},

\begin{itemize}
  \setlength\itemsep{0.25em}
\item in the case that $len_{i} \leq d_{i}-\tau_{L-m}$, $\beta_{i} \leftarrow D_{j}$;

\item otherwise, $\beta_{i} \leftarrow k_{i}\cdot(d_{i}-\tau_{L-m})$.
\end{itemize}
\end{itemize}
Here, $\beta_{i}$ represents the maximum workload of a task $T_{i}$ that could be executed in $[\tau_{L-m}+1, \tau_{L}]$.
\end{definition}

Built on Definition~\ref{Def-1}, we move to the case where $C$ is finite and define the maximum amount of resource $\lambda_{m}^{C}(\mathcal{S})$ that can be utilized by $\mathcal{S}$ on $C$ machines in every $[\tau_{L-m}+1, \tau_{L}]$, $m\in [1, L]$.


\begin{figure}[!ht]
  \centering

  \includegraphics[width=3.45in]{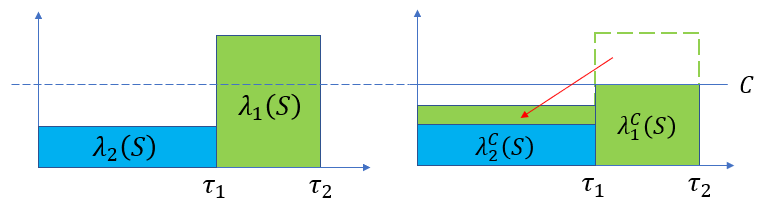}

  \caption{Derivation from the definition $\lambda_{m}(\mathcal{S})$ to $\lambda_{m}^{C}(\mathcal{S})$.}\label{Fig.2}
\end{figure}


To help readers grasp the underlying intuition in the process of deriving $\lambda_{m}^{C}(\mathcal{S})$ from $\lambda_{m}(\mathcal{S})$, we first illustrate this process in the case where $L=2$ with the help of Fig.~\ref{Fig.2}. Fig.~\ref{Fig.2} (left) illustrates the parameter $\lambda_{m}(\mathcal{S})$ in Definition~\ref{Def-1}, where the green area denotes $\lambda_{1}(\mathcal{S})$ and the green and blue areas together denote $\lambda_{2}(\mathcal{S})$. As illustrated in Fig.~\ref{Fig.2} (right), due to the capacity constraint that $C$ is finite, we have that
\begin{enumerate}
\item [\textbf{(\rmnum{1})}] $C\cdot(\tau_{2}-\tau_{1})$ is the maximum possible workload that could be processed in $[\tau_{1}+1, \tau_{2}]$ due to the capacity constraint, and $\lambda_{1}(\mathcal{S})$ is the maximum available workload of $\mathcal{S}$ that needs to be processed in $[\tau_{1}+1, \tau_{2}]$ due to the deadline and parallelism constraints. As a result, on $C$ machines, the maximum workload $\lambda_{1}^{C}(\mathcal{S})$ of $\mathcal{S}$ that can be processed in $[\tau_{1}+1, \tau_{2}]$ is the size of the green area in $[\tau_{1}+1, \tau_{2}]$, i.e.,
    \begin{align*}
    \lambda_{1}^{C}(\mathcal{S})  =\min\{C\cdot(\tau_{2}-\tau_{1}), \lambda_{1}(\mathcal{S})\}=C\cdot(\tau_{2}-\tau_{1}).
    \end{align*}

\item [\textbf{(\rmnum{2})}] After $\lambda_{1}^{C}(\mathcal{S})$ workload of $\mathcal{S}$ has been processed in $[\tau_{1}+1, \tau_{2}]$, the remaining workload of $\mathcal{S}$ that needs to processed in $[1, \tau_{1}]$ is $\lambda_{2}(\mathcal{S})-\lambda_{1}^{C}(\mathcal{S})$; the maximum workload that could be processed in $[1, \tau_{1}]$ is $C\cdot \tau_{1}$ due to the capacity constraint. As a result, $\lambda_{2}^{C}(\mathcal{S})$ is defined as follows:
\begin{align*}
& \lambda_{2}^{C}(\mathcal{S})\\
& =\lambda_{1}^{C}(\mathcal{S})+\min\{ C\cdot(\tau_{1}-\tau_{0}),\enskip \lambda_{2}(\mathcal{S})-\lambda_{1}^{C}(\mathcal{S}) \} \\
&=\min\{C\cdot(\tau_{2}-\tau_{0}), \enskip \lambda_{2}(\mathcal{S})\} = \lambda_{2}(\mathcal{S}),
\end{align*}
i.e., the size of all the colored areas in $[\tau_{0}+1, \tau_{2}]$.
\end{enumerate}
Generalizing the above process, we derived a recursive definition of $\lambda_{m}^{C}(\mathcal{S})$.

\begin{definition}\label{Def-2}
In the case where $C$ is finite (i.e., with the capacity constraint), for all $m\in[L]$, the maximum amount of resource $\lambda_{m}^{C}(\mathcal{S})$ that could be utilized by $\mathcal{S}$ in $[\tau_{L-m}+1, \tau_{L}]$ is defined by the following recursive procedure:
\begin{itemize}
  \setlength\itemsep{0.3em}
\item set $\lambda_{0}^{C}(\mathcal{S})$ to zero trivially;
\item set $\lambda_{m}^{C}(\mathcal{S})$ to the sum of $\lambda_{m-1}^{C}(\mathcal{S})$ and \\
$\min\left\{\lambda_{m}(\mathcal{S})-\lambda_{m-1}^{C}(\mathcal{S}),\enskip C\cdot\left(\tau_{L-m+1}-\tau_{L-m}\right) \right\}$.
\end{itemize}
\end{definition}



We finally state our definition that formalizes the concept of optimal utilization of $C$ machines by a set $\mathcal{S}$ of malleable tasks with deadlines:

\begin{definition}[Optimal Resource Utilization State]\label{Def-3}
We say that $C$ machines are optimally utilized by a set of tasks $\mathcal{S}$, if, for all $m\in[L]^{+}$, $\mathcal{S}$ utilizes $\lambda_{m}^{C}(\mathcal{S})$ resources in $[\tau_{L-m}+1, d]$ on $C$ machines.
\end{definition}

We define $\mu_{m}^{C}(\mathcal{S})=\sum_{T_{i}\in\mathcal{S}}{D_{i}}-\lambda_{L-m}^{C}(\mathcal{S})$ as the remaining (minimum) workload of $\mathcal{S}$ that needs to be processed after $\mathcal{S}$ has maximally utilized $C$ machines in $[\tau_{m}+1, \tau_{L}]$ for all $m\in [L-1]$.

\begin{lemma}[Boundary Condition]\label{boundary-condition}
If there exists a feasible schedule for $\mathcal{S}$,  the following inequality holds for all $m\in[L-1]$:
\begin{center}
$\mu_{m}^{C}(\mathcal{S})\leq C\cdot\tau_{m}$,
\end{center}
which is referred to as \textbf{{\em boundary condition}} in this paper.
\end{lemma}
\begin{proof}
Recall the definition of $\lambda_{L-m}^{C}(\mathcal{S})$ in Definition~\ref{Def-2}. After $\mathcal{S}$ has maximally utilized the machines in $[\tau_{m}+1, d]$ and been allocated the maximum amount of resource, i.e., $\lambda_{L-m}^{C}(\mathcal{S})$, if there exists a feasible schedule for $\mathcal{S}$, the total amount of the remaining demands of $\mathcal{S}$ to be processed should be no more than the capacity $C\cdot\tau_{m}$ in $[1, \tau_{m}]$.
\end{proof}

\begin{table}
\centering
\begin{threeparttable}[!ht]

\caption{Main Notation for the algorithms LDF($\mathcal{S}$), Fully-Utilize($i$), Fully-Allocate($i$), and AllocateRLM($i$, $\theta_{1}$, $x$)}

\begin{tabular}{|C{1.2cm}|C{6.2cm}|}
\hline
   Notation & Explanation\\
\hline

$\mathcal{S}$ & a set of tasks to be allocated by LDF($\mathcal{S}$) and $\mathcal{S}\subseteq\mathcal{T}$ \\ \hline

$\mathcal{S}_{i}$ & the tasks of $\mathcal{S}$ with a deadline $\tau_{i}$  \\ \hline

$\lambda_{m}(\mathcal{S})$ & the maximum amount of resource that could be utilized by $\mathcal{S}$ in $[\tau_{L-m}+1, \tau_{L}]$ in an idealized case where there is an indefinite number of machines, $m\in[L]^{+}$   \\ \hline

$\lambda_{m}^{C}(\mathcal{S})$  & the maximum amount of resource that can be utilized by $\mathcal{S}$ on $C$ machines in every $[\tau_{L-m}+1, \tau_{L}]$, $m\in[L]^{+}$  \\ \hline

$\mu_{m}^{C}(\mathcal{S})$  &  the remaining workload of $\mathcal{S}$ that needs to be processed after $\mathcal{S}$ has optimally utilized $C$ machines in $[\tau_{m}+1, \tau_{L}]$, i.e., $\mu_{m}^{C}(\mathcal{S})=\sum_{T_{i}\in\mathcal{S}}{D_{j}}-\lambda_{L-m}^{C}(\mathcal{S})$, $m\in[L-1]$  \\ \hline

$T_{i}$ & a task that is being allocated by the algorithm LDF($\mathcal{S}$); the actual allocation is done by Allocate-B($i$) \\ \hline

$\mathcal{S}^{\prime}$ & so far, all tasks that have been fully allocated by LDF($\mathcal{S}$) and are considered before $T_{i}$  \\ \hline

$\mathcal{S}^{\prime}$ & $\mathcal{S}^{\prime\prime}=\mathcal{S}^{\prime}\cup\{ T_{i} \}$  \\ \hline

$t_{0}$ & a turning point defined in Property~\ref{proper-2}, with time slots respectively later than and no later than $t_{0}$ having different resource utilization state  \\ \hline

$t_{1}$ & similar to $t_{0}$, a turning point defined in Lemma~\ref{lemma-fully-utilize} upon completion of Fully-Utilize($i$)  \\ \hline

$t_{2}$ & similar to $t_{0}$, a turning point defined in Lemma~\ref{lemma-allocate} upon completion of Fully-Allocate($i$)  \\ \hline

$t^{\prime}$ & the latest time slot in $[1, \tau_{m}]$ with $\overline{W}(t^{\prime})>0$  \\ \hline

$t^{\prime\prime}$, $t^{\prime\prime\prime}$ & a time slot that satisfies some property defined and only used in Section~\ref{sec.phase-2}  \\ \hline

\end{tabular}
\label{table-2}
 \end{threeparttable}
\end{table}

\subsection{Scheduling Algorithm}
\label{sec.scheduling}

In this section, we assume that $\mathcal{S}$ satisfies the boundary condition above, and, propose an algorithm LDF($\mathcal{S}$) that achieves the optimal resource utilization state, producing a feasible schedule for $\mathcal{S}$.

\subsubsection{Overview of LDF($\mathcal{S}$)}
\label{sec.overview}

\begin{figure*}
  \centering
  \includegraphics[width=2.2in]{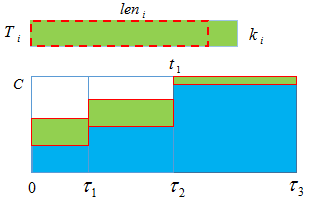}
  \includegraphics[width=2.05in]{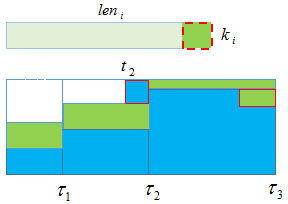}
  \includegraphics[width=2.2in]{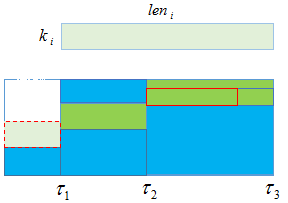}

  \caption{The resource allocation state of $T_{i}$ and the previous tasks $\mathcal{S}^{\prime}$ respectively upon completion of Fully-Utilize($i$), Fully-Allocate($i$), and AllocateRLM($i$, $1$, $t_{2}+1$) where $L=m=3$: the blue area in the rectangle denotes the allocation to the previous tasks that satisfies Property~\ref{proper-1} and Property~\ref{proper-2} before executing Allocate-B($i$) while the green area in the interval $[1, \tau_{3}]$ denotes the allocation to $T_{i}$ at every time slot.}\label{Fig.5}
\end{figure*}


Initially, for all $T_{i}\in\mathcal{S}$ and $t\in[1, d]$, we set the allocation $y_{i}(t)$ to zero and LDF($\mathcal{S}$) runs as follows:
\begin{enumerate}
 \setlength\itemsep{0.15em}
  \item the tasks in $\mathcal{S}$ are considered in the non-increasing order of the deadlines, i.e., in the order of $\mathcal{S}_{L}$, $\mathcal{S}_{L-1}$, $\cdots$, $\mathcal{S}_{1}$;

  \item for a task $T_{i}$ being considered, the algorithm Allocate-B($i$), presented as Algorithm~\ref{Allocate-B}, is called to allocate $D_{i}$ resource to $T_{i}$ under the constraints of deadline and parallelism bound.
\end{enumerate}

At a high level, we show in the following that, only if $\mathcal{S}$ satisfies the boundary condition and the resource utilization satisfies some properties upon every completion of Allocate-B($\cdot$), all tasks in $\mathcal{S}$ will be fully allocated.

Now, we begin to elaborate this high-level idea. In LDF($\mathcal{S}$), when a task $T_{i}$ is being considered, suppose that the allocated task $T_{i}$ belongs to $\mathcal{S}_{m}$ and denote by $\mathcal{S}^{\prime}\subseteq\mathcal{S}_{L}\cup\cdots\cup\mathcal{S}_{m}$ the tasks that have been fully allocated so far and are considered before $T_{i}$. Here, $\mathcal{S}$ satisfies the boundary condition and so do all its subsets including $\mathcal{S}^{\prime}$ and $\mathcal{S}^{\prime}\cup\{ T_{i} \}$.  Before the execution of Allocate-B($i$), we assume that the resource utilization satisfies the following two properties:


Recall the optimal resource utilization state in Definitions~\ref{Def-3}, and the first property is that such an optimal resource utilization state of $C$ machines is achieved by the current allocation to $\mathcal{S}^{\prime}$.
\begin{property}\label{proper-1}
For all $l\in[L]^{+}$, $\mathcal{S}^{\prime}$ is allocated $\lambda_{l}^{C}(\mathcal{S}^{\prime})$ resource in $[\tau_{L-l}+1, d]$ where $\lambda_{l}^{C}(\mathcal{S}^{\prime})$ is defined in Definition~\ref{Def-2}.
\end{property}

The second property is that a stepped-shape resource utilization state is achieved in $[1, \tau_{m}]$ by the current allocation to $\mathcal{S}^{\prime}$.
\begin{property}\label{proper-2}
If there exists a time slot $t\in [1,\tau_{m}]$ such that $\overline{W}(t)>0$, let $t_{0}$ be the latest slot in $[1,\tau_{m}]$ such that $\overline{W}(t^{\prime})>0$; then we have $\overline{W}(1)\geq \overline{W}(2)\geq \cdots \geq \overline{W}(t_{0})$.
\end{property}

If Property~\ref{proper-1} and Property~\ref{proper-2} hold, we will show in Section~\ref{sec.phase-1} and \ref{sec.phase-2} that, there exists an algorithm Allocate-B($i$) such that, upon completion of Allocate-B($i$), the following two properties are satisfied:

\begin{property}\label{proper-3}
$T_{i}$ is fully allocated.
\end{property}

\begin{property}\label{proper-4}
The resource allocation to $\mathcal{S}^{\prime}\cup\{T_{i}\}$ still satisfies Property~\ref{proper-1} and Property~\ref{proper-2}.
\end{property}
Due to the existence of the above Allocate-B($i$), only if $\mathcal{S}$ satisfies the boundary condition, $\mathcal{S}$ can be fully allocated by LDF($\mathcal{S}$). The reason for this can be explained by induction. When the first task $T_{i}$ in $\mathcal{S}$ is considered, $\mathcal{S}^{\prime}$ is empty, and, before the execution of Allocate-B($i$), Property~\ref{proper-1} and Property~\ref{proper-2} holds trivially. Further, upon completion of Allocate-B($i$), $T_{i}$ will be fully allocated by Allocate-B($i$) due to Property~\ref{proper-3}, and Property~\ref{proper-4} still holds. Then, assume that $\mathcal{S}^{\prime}$ that denotes the current fully allocated tasks is nonempty and Property~\ref{proper-1} and Property~\ref{proper-2} hold; the task $T_{i}$ being considered by LDF($\mathcal{S}$) will still be fully allocated and Property~\ref{proper-3} and Property~\ref{proper-4}, upon completion of Allocate-B($i$). Hence, all tasks in $\mathcal{S}$ will be finally fully allocated upon completion of LDF($\mathcal{S}$).

In the rest of this subsection, we will propose an algorithm Allocate-B($i$) mentioned above such that, upon completion of Allocate-B($i$), Property~\ref{proper-3} and Property~\ref{proper-4} holds, {\em if}, before the execution of Allocate-B($i$), the resource allocation to $\mathcal{S}^{\prime}$ satisfies Property~\ref{proper-1} and Property~\ref{proper-2} hold. Then, we immediately have the following proposition:
\begin{proposition}\label{property-task}
If $\mathcal{S}$ satisfies the boundary condition, LDF($\mathcal{S}$) will produce a feasible schedule of $\mathcal{S}$ on $C$ machines.
\end{proposition}


\vspace{0.05em}\noindent\textbf{Overview of Allocate-B($i$).} The construction of Allocate-B($i$) will proceed with two phases. In the first phase, we introduce what operations are feasible to make $T_{i}$ fully allocated $D_{i}$ resource under Property~\ref{proper-1} and Property~\ref{proper-2}. We will use two algorithms Fully-Utilize($i$) and Fully-Allocate($i$) to describe them, and the sketch of this phase is as follows:
\begin{itemize}
 \setlength\itemsep{0.2em}

\item From the deadline $d_{i}$ towards earlier time slots, Fully-Utilize($i$) makes $T_{i}$ fully utilize the maximum amount of machines available at every slot. Upon its completion, the resource allocation state is illustrated in Fig.~\ref{Fig.5} (left) and will be described in Lemma~\ref{lemma-fully-utilize}.

\item If $T_{i}$ is not fully allocated yet, as illustrated in Fig.~\ref{Fig.5} (middle), Fully-Allocate($i$) transfers the allocation of the previous tasks $\mathcal{S}^{\prime}$ at the time slots closest to $d_{i}$ to the latest slots in $[1, d_{i}]$ that have idle machines, so that, $k_{i}$ machines are finally allocated to $T_{i}$ at each of these slots closest to $d_{i}$; as a result, $T_{i}$ is fully allocated.

\end{itemize}
Upon completion of Fully-Allocate($i$), the resource allocation state may not satisfy Property~\ref{proper-1} and Property~\ref{proper-2}, as illustrated by Fig.~\ref{Fig.5} (middle). We therefore propose an algorithm AllocateRLM($i$, $\eta_{1}$, $x$) in the second phase:
\begin{itemize}
\item the allocation of the previous tasks at every slot $t$ closest to the deadline is again transferred to the latest slots that have idle machines, and, the allocation of $T_{i}$ in the earliest slots is transferred to $t$; the final resource allocation state is illustrated in Fig.~\ref{Fig.5} (right).
\end{itemize}

Following the above high-level ideas, the details of the first and second phases are respectively presented in Section~\ref{sec.phase-1} and Section~\ref{sec.phase-2}.

\subsubsection{Phase 1}
\label{sec.phase-1}





Now, we introduce Fully-Utilize($i$) and Fully-Allocate($i$) formally. Before their execution, recall that we assume in the last subsection $T_{i}\in\mathcal{S}_{m}$; the allocation to the previously allocated tasks in $\mathcal{S}^{\prime}$ satisfies Properties~\ref{proper-1} and~\ref{proper-2}. The whole set of tasks $\mathcal{S}$ to be scheduled satisfies the boundary condition where $\mathcal{S}^{\prime}\varsubsetneq\mathcal{S}$.

Initially, set $y_{i}(t)$ to zero for all time slots, and, \textbf{Fully-Utilize($i$)} operates as follows:
\begin{itemize}
\item for every time slot $t$ from the deadline $d_{i}$ to 1, set $y_{i}(t)\leftarrow\min\{k_{i}, D_{i}-\sum_{\overline{t}=t+1}^{d_{i}}{y_{i}(\overline{t})}, \overline{W}(t)\}$.
\end{itemize}
Here, $k_{i}$ is the parallelism bound, $D_{i}-\sum_{\overline{t}=t+1}^{d_{i}}{y_{i}(\overline{t})}$ is the remaining workload to be processed upon completion of its allocations at slots $t+1, \cdots, d_{i}$, and $\overline{W}(t)$ is the number of machines idle at $t$; specially, $\sum_{\overline{t}=d_{i}+1}^{d_{i}}{y_{i}(\overline{t})}$ is set to 0, representing the allocation to $T_{i}$ is zero before the allocation begins. Their minimum denotes the maximum amount of machines that $T_{i}$ can or needs to utilize at $t$ after the allocation to $T_{i}$ at slots $t+1,\cdots, d_{i}$.

Before executing Fully-Utilize($i$), the resource allocation to the previous tasks $\mathcal{S}^{\prime}$ satisfies Property~\ref{proper-1}. Its execution does not change the previous allocation to $\mathcal{S}^{\prime}$. Let
\begin{center}
$\mathcal{S}^{\prime\prime}=\mathcal{S}^{\prime}\cup\{T_{i}\}$.
\end{center}
Since $d_{i}=\tau_{m}$, the workload of $T_{i}$ can only be processed in $[1, \tau_{m}]$; the maximum workload of $\mathcal{S}^{\prime\prime}$ that could be processed in $[\tau_{m}+1, \tau_{L}]$ still equals its counterpart when $\mathcal{S}^{\prime}$ is considered. We come to the following conclusion in order to not violate the boundary condition:


\begin{lemma}\label{exit-condition-1}
Upon completion of Fully-Utilize($i$), all tasks of $\mathcal{S}$ would have been fully allocated {\em in the case that} the total allocation to $\mathcal{S}^{\prime\prime}$ in $[1, \tau_{m}]$ is $C\cdot\tau_{m}$, i.e., $C\cdot\tau_{m} = \sum_{T_{j}\in\mathcal{S}^{\prime\prime}}{\sum_{\overline{t}=1}^{\tau_{m}}{y_{j}(\overline{t})}}$.
\end{lemma}
\begin{proof}
See the appendix for detailed proof.
\end{proof}




Upon completion of Fully-Utilize($i$), {\em in the other case that} the total allocation to $\mathcal{S}^{\prime\prime}$ is $< C\cdot\tau_{m}$, even for $T_{i}$, it may not be fully allocated. In this case, there exists a slot $t\in [1, \tau_{m}]$ such that $\overline{W}(t)>0$, and let $t_{1}$ denote the latest such time slot in $[1, \tau_{m}]$.

In Fully-Utilize($i$), upon completion of the allocation to $T_{i}$ at $t\in [1, t_{1}]$, if $T_{i}$ has not been fully allocated yet, it is allocated $k_{i}$ machines at $t$, i.e., $y_{i}(t)=k_{i}$, and these allocations in $[1, t_{1}]$ are non-decreasing, i.e.,
\begin{center}
$y_{i}(1)\leq y_{i}(2)\leq \cdots\leq y_{i}(t_{1})$.
\end{center}
Before executing Fully-Utilize($i$), the numbers of idle machines have a stepped shape, i.e., $\overline{W}(1)\geq \cdots\geq \overline{W}(t_{0})$ by Property~\ref{proper-2}, where $\overline{W}(t)=C-W(t)=C-\sum_{T_{j}\in\mathcal{S}^{\prime}}{y_{j}(t)}$. Upon its completion, with $y_{i}(t)$ machines occupied by $T_{i}$, we conclude that

\begin{lemma}\label{lemma-fully-utilize}
Upon completion of Fully-Utilize($i$), in the case that the total allocation to $\mathcal{S}^{\prime\prime}$ is $< C\cdot\tau_{m}$,
\begin{itemize}
 \setlength\itemsep{0.3em}
  \item for all $t\in[1, t_{1}]$, if the total allocation of $T_{i}$ in $[t, d_{i}]$ is $<$ the workload of $T_{i}$, i.e.,  $D_{i}-\sum_{\overline{t}=t}^{d_{i}}{y_{i}(\overline{t})}>0$, we have $y_{i}(t)=k_{i}$;

  \item the numbers of idle/unallocated machines in $[1, t_{1}]$ have a stepped shape, i.e., $\overline{W}(1)\geq \cdots \geq \overline{W}(t_{1})>0$.
\end{itemize}
\end{lemma}


With the current resource allocation state shown in Lemma~\ref{lemma-fully-utilize}, we are enabled to propose the algorithm Fully-Allocate($i$) to make $T_{i}$ fully allocated. Deducting the current resource allocated to $T_{i}$, let $\Omega$ denote the remaining workload of $T_{i}$ to be allocated more resource, i.e.,
\begin{center}
$\Omega=D_{i}-\sum_{\overline{t}\leq d_{i}}{y_{i}(\overline{t})}$.
\end{center}
For every slot $t\in[1, t_{1}]$, the number $y_{i}(t)$ of machines allocated to $T_{i}$ at $t$ is $k_{i}$ in the case that $\Omega>0$ by Lemma~\ref{lemma-fully-utilize}. The total workload $D_{i}$ is $\leq k_{i}\cdot d_{i}$, and, with the parallelism bound, \textbf{Fully-Allocate}($i$) considers each slot $t$ from $d_{i}$ towards $t_{1}+1$ and operates as follows repeatedly at each $t$ until $\Omega = 0$:
\begin{enumerate}
 \setlength\itemsep{0.3em}

  \item   $\Delta\leftarrow \min\{k_{i}-y_{i}(t), \Omega\}$.

          \setlength{\parindent}{1em}\vspace{0.15em}{\em Notes.} $k_{i}-y_{i}(t)$ is the maximum number of additional machines that could be utilized at $T$ with its previous allocation $y_{i}(t)$. 

  \item Call Routine($\Delta$, 1, 0, $t$), presented as Algorithm~\ref{Routine}.

        \setlength{\parindent}{1em}\vspace{0.15em}{\em Notes.} Routine($\cdot$) aims to increase the number of available machines $\overline{W}(t)$ at $t$ to $\Delta$ by transferring the allocation of other tasks to an earlier time slot.

  \item Allocate $\overline{W}(t)$ more machines to $T_{i}$ at $t$: $y_{i}(t)\leftarrow y_{i}(t)+\overline{W}(t)$, and, $\Omega\leftarrow \Omega-\overline{W}(t)$.

      \setlength{\parindent}{1em}\vspace{0.15em}{\em Notes.} $\Omega$ denotes the currently remaining workload to be processed; in this iteration, if $\Omega>0$ currently, $\Delta=k_{i}-y_{i}(t)$ and the allocation $y_{i}(t)$ of $T_{i}$ at $t$ becomes $k_{i}$.

  \item $t\leftarrow t-1$.
\end{enumerate}

\begin{algorithm}[!ht]
\SetKwInOut{Begin}{Begin}
\SetKwInOut{Input}{Input}


  \While{$\overline{W}(t)<\Delta$}{
	
    $t^{\prime}\leftarrow$  the current time slot earlier than and closest to $t$ so that $\overline{W}(t^{\prime})>0$\;

    \If{$\eta_{1}=1$}{
      \If{there exists no such $t^{\prime}$}{
        $flag \leftarrow 1$, break\;
      }
    }
    \Else{
      \If{$t^{\prime}\leq t_{m}^{th}$, or there exists no such $t^{\prime}$}{
        $flag \leftarrow 1$, break\;
      }
    }
	
	\If{$\eta_{2}=1$}{
        \If{$\sum_{\overline{t}=1}^{t^{\prime}-1}{y_{i}(\overline{t})}\leq \overline{W}(t)$}{
           $flag \leftarrow 1$, break\;
        }
	}

	let $i^{\prime}$ be a task such that $y_{i^{\prime}}(t)>y_{i^{\prime}}(t^{\prime})$\;

    $y_{i^{\prime}}(t)\leftarrow y_{i^{\prime}}(t)-1$, $y_{i^{\prime}}(t^{\prime})\leftarrow y_{i^{\prime}}(t^{\prime})+1$\;

 }
\caption{Routine($\Delta$, $\eta_{1}$, $\eta_{2}$, $t$)\label{Routine}}
\end{algorithm}

Now, we explain the existence of $T_{i^{\prime}}$ in line 12 of Routine($\cdot$) and the reason why $T_{i}$ will be finally fully allocated by Fully-Allocate($i$). The only operation that changes the allocation to $T_{i}$ occurs at the third step of Fully-Allocate($i$). Hence, we have
\begin{lemma}\label{lemma-decrease}
Fully-Allocate($i$) never decreases the allocation $y_{i}(t)$ to $T_{i}$ at any time slot $t\in [1,d_{i}]$ during its execution, compared with the $y_{i}(t)$ just before executing Fully-Allocate($i$).
\end{lemma}

We could also prove by contradiction that
\begin{lemma}\label{lemma-existence}
When Routine($\Delta$, 1, 0, $t$) is called, the task $T_{i^{\prime}}$ in line 12 always exists if (\rmnum{1}) the condition in line 4 is false, (\rmnum{2}) $y_{i}(t^{\prime})=k_{i}$, and (\rmnum{3}) $y_{i}(t)<k_{i}$ and $\overline{W}(t)=0$.
\end{lemma}
\begin{proof}
See the Appendix for the detailed proof.
\end{proof}

At each iteration of Fully-Allocate($i$), if there exists a $t^{\prime}$ such that $\overline{W}(t^{\prime})>0$ in the loop of Routine($\cdot$), with Lemmas~\ref{lemma-fully-utilize} and~\ref{lemma-decrease}, we have $y_{i}(t^{\prime})=k_{i}$. Since $\Omega>0$ and $y_{i}(t) < k_{i}$, when Routine($\cdot$) is called, we have $\overline{W}(t)=0$; otherwise, this contradicts Lemma~\ref{lemma-fully-utilize}. With Lemma~\ref{lemma-existence}, we will conclude that the task $T_{i^{\prime}}$ in line 11 exists when it is called by Fully-Allocate($i$). In addition, the operation at line 12 of Routine($\cdot$) does not change the total allocation to $T_{i^{\prime}}$, and violate the parallelism bound $k_{i^{\prime}}$ of $T_{i^{\prime}}$ since the current $y_{i^{\prime}}(t^{\prime})$ is no more than the initial $y_{i^{\prime}}(t)$.

\begin{proposition}\label{proposi-fully-allocate}
Upon completion of Fully-Allocate($i$), the task $T_{i}$ is fully allocated.
\end{proposition}
\begin{proof}
Fully-Allocate($i$) ends up with one of the following three events. The first is that the condition in line 4 of Routine($\cdot$) is true. Then, with Lemma~\ref{exit-condition-1}, all tasks in $\mathcal{S}$ has been fully allocated. If the first event doesn't happen, the second is $\Omega=0$ and $T_{i}$ has been fully allocated. If the first and second events don't happen, the third occurs after finishing the iteration of Fully-Allocate($i$) at time slot $t_{1}+1$; then, there is a slot $t^{\prime}$ in $[1, t_{1}+1]$ that are not fully utilized. As a result, we have that $T_{i}$ has been fully allocated; otherwise, $\Omega>0$, which implies $y_{i}(t_{1}+1)=k_{i}$, and we have $y_{i}(t)=k_{i}$ for all $t\in[d_{i}]^{+}$ due to Lemma~\ref{lemma-fully-utilize}, which contradicts $\Omega>0$. Finally, the theorem holds.
\end{proof}

Upon completion of Fully-Utilize($i$), the resource allocation feature is described in Lemma~\ref{lemma-fully-utilize} and illustrated in Fig.~\ref{Fig.5} (left). Built on this, Fully-Allocate($i$) considers every slot from $d_{i}$ to $t_{1}+1$; as illustrated in Fig.~\ref{Fig.5} (middle) and roughly explained there, upon completion of Fully-Allocate($i$), the resource allocation feature is described as follows.

\begin{lemma}\label{lemma-allocate}
Upon completion of Fully-Allocate($i$), if there exists a $t\in [1, \tau_{m}]$ such that $\overline{W}(t)>0$, let $t_{2}$ be the latest such slot:
\begin{itemize}
 \setlength\itemsep{0.2em}

  \item for all $t\in[1, t_{2}]$, if the total allocation of $T_{i}$ in $[t, d_{i}]$ is $<D_{i}$ (i.e., $D_{i}-\sum_{\overline{t}=t}^{d_{i}}{y_{i}(\overline{t})}>0$), we have $y_{i}(t)=k_{i}$;

  \item the numbers of available machines in $[1, t_{2}]$ have a stepped shape, i.e, $\overline{W}(1)\geq \cdots \geq \overline{W}(t_{2})>0$.
\end{itemize}
Here $t_{2}\geq t_{1}$.
\end{lemma}
\begin{proof}
See the Appendix for the formal proof.
\end{proof}

\subsubsection{Phase 2}
\label{sec.phase-2}

Now, we introduce AllocateRLM($i$, $\eta_{1}$, $x$). Recall that $t^{\prime}$ always denotes the slot closest to but earlier than $\tau_{m}$ (i.e., the latest slot in $[1, \tau_{m}]$) such that $\overline{W}(t^{\prime})>0$ and, before executing AllocateRLM($\cdot$), $t^{\prime}=t_{2}$ due to Lemma~\ref{lemma-allocate}. The resource allocation feature before executing AllocateRLM($i$, $\eta_{1}$, $x$) is described in Lemma~\ref{lemma-allocate} and illustrated in Fig.~\ref{Fig.5} (middle); the underlying intuition of AllocateRLM($i$, $\eta_{1}$, $x$) is described in Section~\ref{sec.overview} and, upon its completion, the resource allocation feature is illustrated in Fig.~\ref{Fig.5} (right).

Formally, \textbf{AllocateRLM($i$, $\eta_{1}$, $x$)} considers each slot $t$ from $d_{i}$ to $x$ and operates as follows repeatedly at each $t$ until the total allocation of $T_{i}$ in $[1, t-1]$, i.e., $\sum_{\overline{t}=1}^{t-1}{y_{i}(\overline{t})}$, equals zero, where $\eta_{1}=1$ and $x=t_{2}+1$ in this section:
\begin{enumerate}
 \setlength\itemsep{0.1em}

  \item  $\Delta\leftarrow \min\{k_{i}-y_{i}(t), \sum_{\overline{t}=1}^{t-1}{y_{i}(\overline{t})}\}$.

        \setlength{\parindent}{1em}\vspace{0.01em}{\em Notes.} $\Delta$ denotes the maximum allocation of $T_{i}$ before $t$ that can be transferred to $t$ with the parallelism constraint.

  \item if $\Delta=0$, go to the step 5; otherwise, execute the steps 3-5.

  \item set $flag \leftarrow 0$ and call Routine($\Delta$, $\eta_{1}$, 1, $t$).

      \setlength{\parindent}{1em}\vspace{0.01em}{\em Notes.} Routine($\cdot$) aims to increase the number $\overline{W}(t)$ of available machines at $t$ to $\Delta$. With Lemma~\ref{lemma-allocate}, the slots $t^{\prime}$ earlier than but closest to $t_{2}+1$ in Routine($\cdot$) will become fully utilized one by one and, together with the next step 4, upon completion of the iteration at $t$, for all $\overline{t}\in \left[ t^{\prime}+1, d_{i} \right]$, $\overline{W}(\overline{t})=0$.
	
  \item set $\theta\leftarrow\overline{W}(t)$. Allocate $\theta$ more machines to $T_{i}$:
  \begin{center}
  $y_{i}(t)\leftarrow y_{i}(t)+\overline{W}(t)$,
  \end{center}
   and reduce the allocations of $T_{i}$ at the earliest slots by $\theta$: in particular, let $t^{\prime\prime}$ be such a slot that $\sum_{\overline{t}=1}^{t^{\prime\prime}-1}{y_{i}(\overline{t})}<\theta$ and $\sum_{\overline{t}=1}^{t^{\prime\prime}}{y_{i}(\overline{t})}\geq \theta$, and execute the following operations:
      \begin{enumerate}
 \setlength\itemsep{0.15em}
        \item set $\theta\leftarrow \theta-\sum_{\overline{t}=1}^{t^{\prime\prime}-1}{y_{i}(\overline{t})}$, and, for every $\overline{t} \in [ 1, t^{\prime\prime}-1]$, $y_{i}(\overline{t})\leftarrow 0$;
        \item $y_{i}(t^{\prime\prime})\leftarrow y_{i}(t^{\prime\prime})-\theta$.
      \end{enumerate}

      \setlength{\parindent}{1em}{\em Notes.} The number of idle machines at $t$ becomes zero again, i.e., $\overline{W}(t)=0$. The allocation $y_{i}(t)$ of $T_{i}$ at every $\overline{t}\in [1, t^{\prime\prime}-1]$ is zero.

  \item if Routine($\Delta$, $\eta_{1}$, 1, $t$) does not change the value of $flag$, i.e., $flag = 0$, $t\leftarrow t-1$; otherwise, exit AllocateRLM($i$, $\eta_{1}$, $x$).

\end{enumerate}
Here, at each slot $t$, when Routine($\cdot$) is called, $\Delta>0$, and $y_{i}(t)<k_{i}$. Further, we have $\overline{W}(t)=0$; otherwise, this contradicts Lemma~\ref{lemma-allocate}. Hence, with Lemma~\ref{lemma-existence}, we conclude that the task $T_{i^{\prime}}$ in line 12 of Routine($\cdot$) exists.

Based on our notes in the description of AllocateRLM($\cdot$), we conclude that
\begin{proposition}\label{theo-structure}
Upon completion of AllocateRLM($i$, 1, $x$) where $x=t_{2}+1$, the final allocation to $\mathcal{S}^{\prime\prime}$ can guarantee that Property~\ref{proper-4} holds where $\mathcal{S}^{\prime\prime} = \mathcal{S}^{\prime} \cup \{T_{i}\}$.
\end{proposition}
\begin{proof}
Fully-Utilize($i$), Fully-Allocate($i$) and AllocateRLM($i$, $\eta_{1}$, $x$) never change the allocation at any slot in $[\tau_{m}+1, d]$. AllocateRLM($i$, 1, $x$) ends up with one of the following four events. The first event occurs when the condition in line 4 of Routine($\cdot$) is true; then, the proposition holds trivially since all the slots $\overline{t}\in [1, d_{i}]$ have been fully utilized, i.e., $\overline{W}(\overline{t})=0$. If the first event doesn't occur, the second event is that, for the first time, at some $t\in[t_{2}+1, d_{i}]$, $\sum_{\overline{t}=1}^{t-1}{y_{i}(\overline{t})}=0$; then, we have that, $T_{i}$ is fully allocated $D_{i}$ resource in $[t, d_{i}]\subseteq [t_{2}+1, d_{i}]$. The third event occurs when the condition in line 10 of Routine($\cdot$) is true. In the following, we will analyze the resource utilization state when either of the second and third events occurs.

Recall that $t^{\prime}$ is defined in line 2 of Routine($\cdot$) where each slot in $[t^{\prime}+1, d_{i}]$ will be fully utilized; when the second or third event occurs, all the slots in $[t^{\prime}+1, d_{i}]$ are fully utilized, i.e., $\overline{W}(\overline{t})=0$, for all $\overline{t}\in [t^{\prime}+1, d_{i}]$. Upon completion of the iteration of AllocateRLM($\cdot$) at $t$ when the third event occurs, or, at $t+1$ when the second event occurs, we have the following three points, in contrast to the allocation achieved just before executing Allocate-B($i$),
\begin{enumerate}
 \setlength\itemsep{0.1em}
\item [(\rmnum{1})] $t^{\prime}\in [1, t_{2}]$ and the allocation to the previous tasks $\mathcal{S}^{\prime}$ at every $\overline{t}\in [1, t^{\prime}-1]$ is still the allocation achieved before executing Allocate-B($i$);

\item [(\rmnum{2})] $\sum_{\overline{t}=1}^{t^{\prime}-1}{y_{i}(\overline{t})}=0$, i.e., the allocation to $T_{i}$ in $[1, t^{\prime}-1]$ is zero and $T_{i}$ is fully allocated $D_{i}$ resource in $[t^{\prime}, d_{i}]$;

\item [(\rmnum{3})] the allocation to $\mathcal{S}^{\prime}$ at $t^{\prime}$ is not decreased;

\item [(\rmnum{4})] the allocation to $T_{i}$ at $t^{\prime}$ does not change.
\end{enumerate}
Noticing the above resource allocation state in $[1, d_{i}]$ where $d_{i}=\tau_{m}$, since Property~\ref{proper-2} holds before executing Allocate-B($i$), we conclude that Property~\ref{proper-2} still holds upon its completion where $t_{0}=t^{\prime}$. Without loss of generality, assume that $t^{\prime}\in [\tau_{m^{\prime}-1}+1, \tau_{m^{\prime}}]$ for some $m^{\prime}\in [m]^{+}$. Then, all the slots in $[\tau_{m^{\prime}}+1, d_{i}]$ have been fully utilized and the allocation in $[\tau_{m}+1, d]$ does not change at all; hence, we have that every interval $[\tau_{l}+1, d]$, where $m^{\prime}\leq l\leq L$, is optimally utilized by $\mathcal{S}^{\prime}\cup\{T_{i}\}$ due to Property~\ref{proper-1}. Since the total allocation to $\mathcal{S}^{\prime}$ in $[1, \tau_{m^{\prime}-1}]$ isn't changed by Allocate-B($i$) if $m^{\prime}-1>0$, due to Property~\ref{proper-1}, the interval $[\tau_{m^{\prime}-1}+1, d]$ is still optimally utilized by $\mathcal{S}^{\prime}$ and the task $T_{i}$ is fully allocated $D_{i}$ resource in this interval; hence, it is still optimally utilized by $\mathcal{S}^{\prime}\cup\{T_{i}\}$. Further, every interval $[\tau_{l}+1, d]$ is also optimally utilized where $1\leq l\leq m^{\prime}-1$. Hence, the theorem holds.

If the first three events don't occur, the fourth event occurs upon completion of the iteration of AllocateRLM($\cdot$) at $t=t_{2}+1$, i.e., the last iteration. In this case, we have that the conditions in lines 4 and 10 of Routine($\cdot$) are always false where at each iteration of AllocateRLM($\cdot$) there always exists such $t^{\prime}$ (defined in line 2 of Routine($\cdot$) with $\overline{W}(t^{\prime})>0$); due to the current resource allocation state, we conclude that, at each of the slots in $[t_{2}+1, d_{i}]$, $T_{i}$ is allocated $k_{i}$ machines. Upon completion of AllocateRLM($\cdot$), there exists a $t^{\prime}$ defined in line 2 of Routine($\cdot$), and, let $t^{\prime\prime\prime}$ denote the earliest slot at which $y_{i}(t^{\prime\prime\prime})\neq 0$ where $t^{\prime\prime\prime} \leq  t^{\prime}$; then, similar to our conclusion in the second and third events, we have that
\begin{enumerate}
 \setlength\itemsep{0.1em}
\item [(\rmnum{1})] the first point here is the same as the first and third points in the last paragraph;
\item [(\rmnum{2})] $T_{i}$ is fully allocated $D_{i}$ resource in $[t^{\prime\prime\prime}, d_{i}]$;

\item [(\rmnum{3})] if $t^{\prime\prime\prime} > t^{\prime}$, the allocation to $T_{i}$ at each $\overline{t}\in[t^{\prime\prime\prime}+1, t^{\prime}]$ does not change and $y_{i}(\overline{t})=k_{i}$ due to Lemma~\ref{lemma-allocate}, and, the allocation to $T_{i}$ at $t^{\prime\prime\prime}$ is greater than zero.
\end{enumerate}
Similar to our analysis in the last paragraph for other events, we conclude that the proposition holds.
\end{proof}

Proposition~\ref{proposi-fully-allocate} and Proposition~\ref{theo-structure} finish to show that Allocate-B($i$) satisfies Property~\ref{proper-3} and Property~\ref{proper-4} and hence completes the proof of Proposition~\ref{property-task}. We finally analyze the time complexity of Allocate-B($i$).

\begin{lemma}\label{lemma-complexity}
The time complexity of Allocate-B($\cdot$) is $\mathcal{O}(n)$.
\end{lemma}
\begin{proof}
See the Appendix for the proof.
\end{proof}

\begin{algorithm}[!ht]
\SetKwInOut{Begin}{Begin}
\SetKwInOut{Input}{Input}

Fully-Utilize($i$)\;

Fully-Allocate($i$)\;

AllocateRLM($i$, 1, $t_{2}+1$)\;

\caption{Allocate-B($i$)\label{Allocate-B}}
\end{algorithm}

Since LDF($\mathcal{S}$) considers a total of $n$ tasks, its complexity is $\mathcal{O}(n^{2})$ with Lemma~\ref{lemma-complexity}. Finally, we draw a main conclusion in this section from Lemma~\ref{boundary-condition} and Proposition~\ref{property-task}:

\begin{theorem}\label{main-theorem}
A set of tasks $\mathcal{S}$ can be feasibly scheduled and be completed by their deadlines on $C$ machines {\em if and only if} the boundary condition holds, where the feasible schedule of $\mathcal{S}$ could be produced by LDF($\mathcal{S}$) with a time complexity $\mathcal{O}(n^{2})$.
\end{theorem}

In other words, if LDF($\mathcal{S}$) cannot produce a feasible schedule for $\mathcal{S}$ on $C$ machines, then $\mathcal{S}$ cannot be successfully scheduled by any algorithm; as a result, LDF($\mathcal{S}$) is optimal. The relationships between the various algorithms of this paper are illustrated in Fig.~\ref{Fig.3} where GreedyRLM will be introduced in the next section. 

\begin{figure*}[!ht]
  \centering

  \includegraphics[width=4.3in]{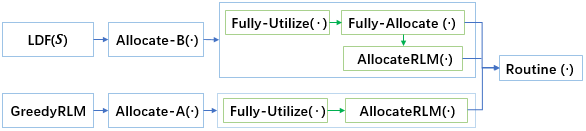}

  \caption{Relationship among Algorithms: for $A\rightarrow B$, the blue and green arrows denote the relations that the algorithm $A$ calls $B$, and, the algorithm $B$ is executed upon completion of $A$.}\label{Fig.3}
\end{figure*}

\vspace{0.2em}\noindent\textbf{Remarks.} We are inspired by the GreedyRTL algorithm \cite{Jain} in the construction of LDF($\cdot$). In terms of the two algorithms themselves, LDF($\cdot$) considers tasks in the decreasing order of deadlines while the order is determined by the marginal values in GreedyRTL($\cdot$). In both algorithms, the allocation to a task $T_{i}$ is considered from $d_{i}$ to 1 (once in GreedyRTL, and possibly three times in LDF($\cdot$)); to make time slots $t$ closest to the deadline of a task $T_{i}$ being considered fully utilized, the key operations are finding a time slot $t^{\prime}$ earlier than $t$ such that there exists a task $T_{i^{\prime}}$ with $y_{i^{\prime}}(t)>y_{i^{\prime}}(t^{\prime})$ when $\overline{W}(t)$, and transferring a part of the allocation of $T_{i^{\prime}}$ at $t$ to $t^{\prime}$. In GreedyRTL($\cdot$), the existence of $T_{i^{\prime}}$ requires that (\rmnum{1}) the number $\overline{W}(t^{\prime})$ of available machines at $t^{\prime}$ is $\geq k$ and (\rmnum{2})\footnote{The particular condition there is $\overline{W}(t) < \min\{ k_{i}, D_{i}-\sum_{\overline{t}=t+1}^{d_{i}}{y_{i}(\overline{t})} \}$.} $\overline{W}(t) < k_{i}$; as a result, before doing any allocation to $T_{i}$ at $t$, the existence could be proved by contradiction. In LDF($\cdot$), to achieve the optimality of resource utilization, one requirement for such existence is relaxed to be that the number of available machines at $t^{\prime}$ is $\geq 1$. The existence is guaranteed by (\rmnum{1}) first make every time slot from $d_{i}$ to 1 fully utilized, as what Fully-Utilize($i$) does, and (\rmnum{2}) a stepped-shape resource utilization state in $[1, d_{i}]$ upon completion of the allocation to the last task, as described in Property~\ref{proper-2}.


\section{Applications: Part \Rmnum{1}}
\label{more-application}

In this section, we illustrate the application of the results in Section~\ref{sec.optimal} to the greedy algorithm for social welfare maximization.


In terms of the maximization problem, {\em the general form of a greedy algorithm} is as follows \cite{Brassard,Even}: it tries to build a solution by iteratively executing the following steps until no item remains to be considered in a set of items: (1) selection standard: in a greedy way, choose and consider an item that is locally optimal according to a simple criterion at the current stage; (2) feasibility condition: for the item being considered, accept it if it satisfies a certain condition such that this item constitutes a feasible solution together with the tasks that have been accepted so far under the constraints of this problem, and reject it otherwise. Here, an item that has been considered and rejected will never be considered again. The selection criterion is related to the objective function and constraints, and is usually the ratio of 'advantage' to 'cost', measuring the efficiency of an item. In the problem of this paper, the constraint comes from the capacity to hold the chosen tasks and the objective is to maximize the social welfare; therefore, the selection criterion here is the ratio of the value of a task to its demand that will refer to as the {\em marginal value} of this task.

Given the general form of greedy algorithm, we define a class GREEDY of algorithms that operate as follows:
\begin{enumerate}
  \item Considers the tasks in the non-increasing order of the marginal value; assume without loss of generality that $v_{1}^{\prime}\geq v_{2}^{\prime}\geq  \cdots \geq v_{n}^{\prime}$;
  \item Denoting by $\mathcal{A}$ the set of the tasks that have been accepted so far, a task $T_{i}$ being considered is accepted and fully allocated {\em iff} there exists a feasible schedule for $\mathcal{A}\cup\{T_{i}\}$.
\end{enumerate}
In the following, we refer to the generic algorithm in GREEDY as {\em Greedy}.

\begin{proposition}\label{bound-greedy}
The best performance guarantee that a greedy algorithm in GREEDY can achieve is $\frac{s-1}{s}$.
\end{proposition}
\begin{proof}
Let us consider a special instance: (\rmnum{1}) let $\mathcal{D}_{i}=\{T_{j}\in\mathcal{T} | d_{i}=d_{i}^{\prime} \}$, where $i\in [2]^{+}$, $d_{2}^{\prime}$ and $d_{1}^{\prime}\in \mathcal{Z}^{+}$, and $d_{2}^{\prime}>d_{1}^{\prime}$; (\rmnum{2}) for all $T_{j}\in\mathcal{D}_{1}$, $v_{j}^{\prime}=1+\epsilon$, $D_{j}=1$,
$k_{j}=1$, and, there is a total of $C\cdot d_{1}^{\prime}$ such tasks, where $\epsilon \in (0, 1)$ is small enough; (\rmnum{3}) for all $T_{j}\in\mathcal{D}_{2}$, $v_{j}^{\prime}=1$, $k_{j}=1$ and $D_{j}=d_{2}^{\prime}-d_{1}^{\prime}+1$. Greedy will always fully allocate resource to the tasks in $\mathcal{D}_{1}$, with all the tasks in $\mathcal{D}_{2}$ rejected to be allocated any resource. The performance guarantee of Greedy will be no more than $\frac{C\cdot d_{1}^{\prime}}{C\cdot[(1+\epsilon)(d_{1}^{\prime}-1)+1\cdot (d_{2}^{\prime}-d_{1}^{\prime}+1)]}$. Further, with $\epsilon\rightarrow 0$, this performance guarantee approaches $\frac{d_{1}^{\prime}}{d_{2}^{\prime}}$. In this instance, $s=\frac{d_{2}^{\prime}}{d_{2}^{\prime}-d_{1}^{\prime}+1}$ and $\frac{s-1}{s}=\frac{d_{1}^{\prime}-1}{d_{2}^{\prime}}$. When $d_{2}^{\prime}\rightarrow +\infty$, $\frac{d_{1}^{\prime}}{d_{2}^{\prime}}=\frac{s-1}{s}$. Hence, the proposition holds.
\end{proof}

\subsection{Notation}

Greedy will consider tasks sequentially. The first considered task will be accepted definitely and then it will use to the feasibility condition to determine whether or not to accept or reject the next task according to the current available resource and the characteristics of this task. To describe the process under which Greedy accepts or rejects tasks, we define the sets of consecutive accepted (i.e., fully allocated) and rejected tasks $\mathcal{A}_1, \mathcal{R}_1, \mathcal{A}_2, \cdots$. Specifically, let $\mathcal{A}_{m}=\{T_{i_{m}}, T_{i_{m}}, \cdots, T_{j_{m}}\}$ be the $m$-th set of the adjacent tasks that are accepted by Greedy where $i_{1}=1$ while $\mathcal{R}_{m}=\{ T_{j_{m}+1}, \cdots, T_{i_{m+1}-1} \}$ is the $m$-th set of the adjacent that are rejected tasks following the set $\mathcal{A}_{m}$, where $m\in[K]^{+}$ for some integer $K$.
Integer $K$ represents the last step: in the $K$-th step, $\mathcal{A}_{K} \neq \emptyset$ and $\mathcal{R}_{K}$ can be empty or non-empty. We also denote by $c_{m}$ the maximum deadline of all rejected tasks of $\cup_{l=1}^{m}{\mathcal{R}_{l}}$, i.e.,
\begin{center}
$c_{m}=\max_{T_{i}\in \cup_{l=1}^{m}{\mathcal{R}_{l}}}{\{ d_{i} \}}$,
\end{center}
and by $c_{m}^{\prime}$ the maximum deadline of $\cup_{l=1}^{m}{\mathcal{A}_{l}}$, i.e.,
\begin{center}
$c_{m}^{\prime}=\max_{T_{i}\in \cup_{l=1}^{m}{\mathcal{A}_{l}}}\{ d_{i} \}$.
\end{center}
While the tasks in $\mathcal{A}_{m}\cup\mathcal{R}_{m}$ are being considered, we refer to Greedy as being in the $m$-th phase. Before the execution of Greedy, we refer to it as being in the 0-th phase. Upon completion of the $m$-th phase of Greedy, we define a {\em threshold} parameter $\tth_m$ such that
\begin{enumerate}
 \setlength\itemsep{0.35em}
\item [(\rmnum{1})] if $c_m\ge c^{\prime}_m$, set $\tth_m = c_m$, and
\item [(\rmnum{2})] if $c_m< c^{\prime}_m$, set $\tth_m$ to any time slot in $[c_m, c^{\prime}_m]$.
\end{enumerate}
Here, $d_{i}\leq t_{m}^{th}$ for all $T_{i}\in\cup_{j=1}^{m}{\mathcal{R}_{j}}$. For ease of the subsequent exposition, we let $t_{0}^{th}=0$ and $t_{K+1}^{th}=d$. We also add a dummy time slot 0 but the task $T_{i}\in\mathcal{T}$ can not get any resource there, that is, $y_{i}(0)=0$ forever. We also let $\mathcal{A}_{0}=\mathcal{R}_{0}=\mathcal{A}_{K+1}=\mathcal{R}_{K+1}=\emptyset$. Besides the notation in Section~\ref{sec.model}, the additional key notation used for this section is also summarized in Table~\ref{table-3}.

\begin{table}
\centering
\begin{threeparttable}[!ht]

\caption{Main Notation for Section~\ref{more-application}}

\begin{tabular}{|C{1.1cm}|C{6.8cm}|}

\hline
   Notation & Explanation\\
\hline

$\mathcal{A}_1, \mathcal{R}_1,$ $\mathcal{A}_2,\cdots,$ $\mathcal{R}_{K}$ & the sets of consecutive accepted (i.e., fully allocated) and rejected tasks by Greedy where $\bigcup_{m=1}^{K}{\mathcal{A}_{m}\cup\mathcal{R}_{m}}=\mathcal{T}$ \\ \hline


$c_{m}$ & the maximum deadline of all rejected tasks of $\cup_{l=1}^{m}{\mathcal{R}_{l}}$  \\ \hline

$c_{m}^{\prime}$ & the maximum deadline of $\cup_{l=1}^{m}{\mathcal{A}_{l}}$ \\ \hline

$\tth_m$ & a threshold parameter such that (\rmnum{1}) if $c_m\ge c^{\prime}_m$, set $\tth_m = c_m$, and (\rmnum{2}) if $c_m< c^{\prime}_m$, set $\tth_m$ to any time slot in $[c_m, c^{\prime}_m]$; when introducing GreedyRLM, it will be set to a specific value  \\ \hline

\end{tabular}
\label{table-3}
 \end{threeparttable}
\end{table}

\subsection{A New Algorithmic Analysis}
\label{sec.algo.analysis}


We will show that as soon as the resource allocation done by Greedy satisfies some features, its performance guarantee can be deduced immediately, i.e., the main result of this subsection is Theorem~\ref{spaa}.

For all $m\in[K]^{+}$, upon completion of Greedy, we define the following two features that we want the allocation to $\cup_{j=1}^{m}{\mathcal{A}_{j}}$ to satisfies:
\begin{feature}\label{utilization1}
The total allocation to $\bigcup_{j=1}^{m}{\mathcal{A}_{j}}$ in $[1, t_{m}^{th}]$ is at least $r\cdot C\cdot t_{m}^{th}$, where $r\in[0,1]$.
\end{feature}

\begin{feature}\label{fully}
For each task $T_{i}\in \bigcup_{j=1}^{m}{\mathcal{A}_{j}}$, its maximum amount of demand that can be processed in each $\left[\tth_{l}+1, d \right]$ is processed where $m\leq l\leq K$, i.e.,
\begin{center}
$\sum\limits_{\overline{t}=\tth_{l}+1}^{d_{i}}{y_{i}(\overline{t})}=\min\left\{ D_{i}, k_{i}(d_{i}-\tth_{l}) \right\}$.
\end{center}
\end{feature}

\begin{theorem}\label{spaa}
If Greedy achieves a resource allocation structure that satisfies Feature~\ref{utilization1} and Feature~\ref{fully} for all $m\in [K]^{+}$, it gives an $r$-approximation to the optimal social welfare.
\end{theorem}


In the rest of Section~\ref{sec.algo.analysis}, we prove Theorem~\ref{spaa}; we will first provide an upper bound of the optimal social welfare.

\vspace{0.45em}\noindent\textbf{Proof Overview.} We refer to the original problem of scheduling $\mathcal{A}_1, \mathcal{R}_1, \cdots, \mathcal{A}_K, \mathcal{R}_{K}$ on $C$ machines to maximize the social welfare as \textbf{the MSW-\Rmnum{1} problem}.

In the following, we define a relaxed version of the MSW-\Rmnum{1} problem.
Assume that $\mathcal{R}_{m}^{\prime}$ consists of a single task $T_{m}^{\prime}$ whose deadline is $t_{m}^{th}$, whose size is infinite, and whose marginal value is the largest one of the tasks in $\mathcal{R}_{m}$, denoted by $\overline{v}_{m}^{\prime}$; here, different from the task in $\mathcal{R}_{m}$, we assume that there is no parallelism constraint on $T_{m}^{\prime}$ whose bound is $C$. In addition, partial execution of the task $T_{m}^{\prime}$ and the tasks of $\mathcal{A}_1, \cdots, \mathcal{A}_K$ can yield linearly proportional value, e.g., if a task $T_{i}\in \mathcal{A}_{l}$ is allocated $\sum_{t=1}^{d_{i}}{y_{i}(t)}<D_{i}$ resource by its deadline, a value $(\sum_{t=1}^{d_{i}}{y_{i}(t)}/D_{i})\cdot v_{i}$ will still be added to the social welfare. We refer to the problem of scheduling $\mathcal{A}_1, \mathcal{R}_1^{\prime}, \cdots, \mathcal{A}_K, \mathcal{R}_{K}^{\prime}$ on $C$ machines as \textbf{the MSW-\Rmnum{2} problem}.

\begin{lemma}\label{lemma-upper-bound-1}
The optimal social welfare of the MSW-\Rmnum{2} problem is an upper bound of the optimal social welfare of the MSW-\Rmnum{1} problem.
\end{lemma}
\begin{proof}
See the appendix for the detailed proof.
\end{proof}

Due to Feature~\ref{utilization1}, Feature~\ref{fully}, and the fact that the marginal value of $T_{m}^{\prime}$ is no larger than the ones of the tasks of $\cup_{l=1}^{m}{\mathcal{A}_{l}}$, we derive the following two lemmas:

\begin{lemma}\label{lemma-upper-bound-2}
The following schedule achieves an upper bound of the optimal social welfare of the MSW-\Rmnum{2} problem, ignoring the capacity constraint:
\begin{enumerate}

\item for all tasks of $\mathcal{A}_1, \cdots, \mathcal{A}_{K}$, their allocation is the same as the one achieved by Greedy with Features~\ref{utilization1} and~\ref{fully} satisfied;

\item for all $m\in [K]^{+}$, execute a part of task $T_{m}^{\prime}$ such that the amount of processed workload in $\left[t_{m-1}^{th}+1, t_{m}^{th}\right]$ is $(1-r)\cdot\left(t_{m}^{th}-t_{m-1}^{th}\right)\cdot C$.
\end{enumerate}
\end{lemma}
\begin{proof}
See the appendix for the detailed proof.
\end{proof}



\begin{lemma}\label{lemma-interval-value}
For all $m\in [K]^{+}$, the total value generated by executing the allocation to $T_{1}^{\prime}, \cdots, T_{m}^{\prime}$ is no larger than $\frac{1-r}{r}$ times the total value generated by the allocation to $\cup_{l=1}^{m}{\mathcal{A}_{l}}$ in $[1, t_{m}^{th}]$.
\end{lemma}
\begin{proof}
See the appendix for the detailed proof.
\end{proof}

In the case that $m=K$, the total value from $T_{1}^{\prime}, \cdots, T_{K}^{\prime}$ is no larger than $\frac{1-r}{r}$ times the total value from the allocation to $\cup_{l=1}^{K}{\mathcal{A}_{l}}$ in $[1, t_{K}^{th}]$. Hence, the total value generated by the schedule in Lemma~\ref{lemma-upper-bound-2} is no larger than $1+\frac{1-r}{r}=\frac{1}{r}$ times the total value generated by the allocation to all tasks of $\mathcal{A}_{1}, \cdots, \mathcal{A}_{K}$. By Lemmas~\ref{lemma-upper-bound-2} and~\ref{lemma-upper-bound-1}, Theorem~\ref{spaa} holds.


\subsection{Optimal Algorithm Design}

We now introduce the executing process of the optimal greedy algorithm GreedyRLM, presented as Algorithm~\ref{GreedyRLM}:
\begin{itemize}
 \setlength\itemsep{0.15em}
\item [(1)] considers the tasks in the non-increasing order of the marginal value.

\item [(2)] in the $m$-th phase, for a task $T_{i}$ being considered, if $\sum_{t\leq d_{i}}{\min\{\overline{W}(t), k_{i}\}}\geq D_{i}$, call Allocate-A($i$), presented as Algorithm~\ref{Allocate-A}, where the details on Fully-Utilize($i$) and AllocateRLM($i$, 0, $t_{m}^{th}+2$) can be found in Section~\ref{sec.phase-1} and Section~\ref{sec.phase-2}.

\item [(3)] if the allocation condition is not satisfied, set the threshold parameter $\tth_m$ of the $m$-th phase that is defined by lines 8-15 of Algorithm~\ref{GreedyRLM}.
\end{itemize}

\begin{algorithm}[!ht]
\SetKwInOut{Input}{Input}
\SetKwInOut{Output}{Output}



\Input{$n$ tasks with $type_{i}=\{ v_{i}, d_{i}, D_{i}, k_{j} \}$}
\Output{A feasible allocation of resources to tasks}
\BlankLine



initialize: $y_{i}(t)\leftarrow 0$ for all $T_{i}\in \mathcal{T}$ and $1\leq t\leq d$, $m=0$, $t_{m}^{th}=0$\;

sort tasks in the non-increasing order of the marginal values: $v_{1}^{\prime}\geq v_{2}^{\prime}\geq \cdots \geq v_{n}^{\prime}$\;

$i\leftarrow 1$\;
\While{$i\leq n$}{

  \If{$\sum_{t\leq d_{i}}{\min\{\overline{W}(t), k_{i}\}}\geq D_{i}$}{\nllabel{forins}
    Allocate-A($i$)\tcp*{\footnotesize{in the $(m+1)$-th phase}}
  }
  \Else{
    \If{$T_{i-1}$ has ever been accepted}{
      $m\leftarrow m+1$\tcp*{\footnotesize{in the $m$-th phase, the allocation to $\mathcal{A}_{m}$ was completed; the first rejected task is $T_{j_{m}}=T_{i}$}}
    }

    \While{$\sum_{t\leq d_{i+1}}{\min\{\overline{W}(t), k_{i+1}\}}< D_{i+1}$}{\nllabel{forins}

      $i\leftarrow i+1$\;
    }\tcc{\footnotesize{the last rejected task is $T_{i_{m+1}-1}=T_{i}$ and $\mathcal{R}_{m}=\{T_{j_{m}}, \cdots, T_{i_{m+1}-1}\}$}}

    \If{$c_m\ge c^{\prime}_m$}{

      $\tth_m \leftarrow c_m$\;
    }
    \Else{

      set $\tth_m$ to time slot just before the first time slot $t$ with $\overline{W}(t)>0$ after $c_m$ or to $c^{\prime}_m$ if there is no time slot $t$ with $\overline{W}(t)>0$ in $[c_m, c^{\prime}_m]$\;
    }
  }
  $i\leftarrow i+1$\;
}


\caption{GreedyRLM\label{GreedyRLM}}
\end{algorithm}

\begin{algorithm}[!ht]
\SetKwInOut{Begin}{Begin}
\SetKwInOut{Input}{Input}

Fully-Utilize($i$)\;

\If{$d_{i} \geq t_{m}^{th}+2$}{
   AllocateRLM($i$, 0, $t_{m}^{th}+2$) where $t_{2}=t_{1}$ that are defined in Section~\ref{sec.phase-1}\;
}

\caption{Allocate-A($i$)\label{Allocate-A}}
\end{algorithm}

When the condition in line 5 of GreedyRLM is true, every accepted task can be fully allocated $D_{i}$ resource using Fully-Utilize($i$). The reason for the existence of $T_{i^{\prime}}$ in Routine($\cdot$) is the same as the reason when introducing LDF($\mathcal{S}$) since $\overline{W}(t^{\prime})>0$.

\begin{proposition}\label{proposition-GreedyRLM}
GreedyRLM gives an $\frac{s-1}{s}$-approximation to the optimal social welfare with a time complexity of $\mathcal{O}(n^{2})$.
\end{proposition}


Now, we begin to prove Proposition~\ref{proposition-GreedyRLM}. The time complexity of Allocate-A($i$) depends on AllocateRLM($\cdot$). Using the time complexity analysis of AllocateRLM($\cdot$) in Lemma~\ref{lemma-complexity}, we get that AllocateRLM($\cdot$) has a time complexity of $\mathcal{O}(n)$, and, the time complexity of GreedyRLM is $\mathcal{O}(n^{2})$. Due to Theorem~\ref{spaa}, in the following, we only need to prove that Features~\ref{utilization1}~and~\ref{fully} holds in GreedyRLM where $r=\frac{s-1}{s}$, which is given in Propositions~\ref{proposition-utilization-2}~and~\ref{fully-2}.

The utilization of GreedyRLM is derived mainly by analyzing the resource allocation state when a task $T_{i}$ cannot be fully allocated (the condition in line 5 of GreedyRLM is not satisfied), and we have that
\begin{proposition}\label{proposition-utilization-2}
Upon completion of GreedyRLM, Feature~\ref{utilization1} holds in which $r=\frac{s-1}{s}$.
\end{proposition}
\begin{proof}
See the Appendix for the detailed proof.
\end{proof}


In GreedyRLM, when a task $T_{i}$ is accepted (lines 5 and 6), Allocate-A($\cdot$) is called to make it fully allocated. In Allocate-A($\cdot$), Fully-Utilize($\cdot$) and AllocateRLM($\cdot$) are sequentially called; both of them consider time slots $t$ from the deadline towards earlier ones: (\rmnum{1}) Fully-Utilize($\cdot$) makes $T_{i}$ utilize the remaining (idle) machines at $t$, and it does not change the allocations of the previous tasks; (\rmnum{2}) at every $t$, if $T_{i}$ does not utilize the maximum number of machines it can utilize (i.e., $y_{i}(t)<k_{i}$), AlloacteRLM($\cdot$) (a) transfers the allocations of the previous allocated tasks to an earlier slot that is closest to $t$ but not fully utilized (i.e., with idle machines), and (b) increases the allocation of $T_{i}$ at $t$ to the maximum (i.e., $k_{i}$) and, correspondingly reduce the equal allocations at the earliest slots, ensuring the total allocation to $T_{i}$ does not exceed $D_{i}$. Finally, upon completion of the whole execution of Allocate-A($\cdot$), we have that
\begin{itemize}
\item the number of allocated machines at each slot does not decrease,
\end{itemize}
in contrast to that amount just before executing Allocate-A($\cdot$). For every accepted task $T_{i}\in\mathcal{A}_{m}$, upon completion of Allocate-A($\cdot$), time slot $\tth_{m}+1$ is not fully utilized by the definition of $\tth_{m}$, i.e., $\overline{W}(\tth_{m}+1)>0$. Further, we have that whenever Allocate-A($\cdot$) completes the allocation to a previous task $T_{i}\in \mathcal{A}_{m^{\prime}}$ where $m^{\prime}<m$, $\tth_{m}+1$ is also not fully utilized then. Based on this, we draw the following conclusion.

\begin{lemma}\label{structure-2}
Due to the definition of $\tth_{m}$, we have for all $m\leq j\leq K$ that
\begin{itemize}
 \setlength\itemsep{0.25em}
\item [(1)] $[ \tth_{j}+1, d ]$ is optimally utilized by $\bigcup_{l=1}^{m}{\mathcal{A}_{l}}$ upon completion of the allocation to it using Allocate-A($i$);
\item [(2)] for the total amount of the allocations to $T_{i}$ in the interval $[ \tth_{j}+1, d ]$ just upon completion of Allocate-A($i$), it does not change upon completion of GreedyRLM.
\end{itemize}
\end{lemma}
\begin{proof}
We first prove the first point. Given a $m^{\prime}\in [m]^{+}$, for every $T_{i}\in\mathcal{A}_{m^{\prime}}$, upon completion of Allocate-A($i$), $\overline{W}(\tth_{j}+1)>0$ for all $j\in [m, K]$; based on this, we conclude that, in the case where $d_{i}\geq \tth_{j}+1$, either $\sum_{t=\tth_{j}+1}^{d_{i}}{y_{i}(t)}=D_{i}$ if $d_{i}-\tth_{j}>len_{i}$ or $y_{i}(t)=k_{i}$ for all $t\in [\tth_{j}+1, d_{i}]$ otherwise. The reason for this conclusion is similar to our analysis for the fourth event when proving Proposition~\ref{theo-structure}; here, there always exists a slot $\tth_{j}+1$ that is not fully utilized, i.e., $\overline{W}(\tth_{j}+1)>0$, leading to that the $t^{\prime}$ defined in line 2 of Routine($\cdot$) always exists where $\overline{W}(t^{\prime})>0$.


Now, we prove the second point in Lemma~\ref{structure-2}. For every $l\in [m^{\prime}, K]$, we observe the subsequent execution of Allocate-A($\cdot$) whose input is a task in $\mathcal{A}_{l}$ and could conclude that,
\begin{enumerate}
 \setlength\itemsep{0.2em}

\item upon its completion, the allocations to $T_{i}$ in $[1, \tth_{l}]$ are still the ones before executing Allocate-A($\cdot$);

\item Allocate-A($\cdot$) can only change the allocations of $T_{i}$ in the time range $[t_{l^{\prime}}+1, t_{l^{\prime}+1}]$ where $l^{\prime} \in [l, K]$ and the total amount of allocations in $[t_{l^{\prime}}+1, t_{l^{\prime}+1}]$ upon its completion is still the amount before its execution.
\end{enumerate}
As a result, we have that, upon completion of Allocate-A($i$), every subsequent execution of Allocate-A($\cdot$) never change the total amount of allocations of $T_{i}$ in $[t_{l^{\prime\prime}}+1, t_{l^{\prime\prime}+1}]$ for all $l^{\prime\prime}\in [m, K]$. 

In the following, it suffices to prove the above two points. In the execution of Allocate-A($\cdot$), Fully-Utilize($\cdot$) is first called and it does not change the allocation to the previous tasks; then, AllocateRLM($\cdot$, 0, $\tth_{l}$) is called in which only Routine($\cdot$) (i.e., its lines 12 and 13) in the step 3 can change the allocation to the previous tasks including $T_{i}$. In lines 12 and 13, a previous task $T_{i^{\prime}}$ is found to change its allocations at $t$ and $t^{\prime}$; here, $t^{\prime}$ is defined in lines 2 and 7 of Routine($\cdot$) and $\tth_{l}<t^{\prime}<t$. As a result, Allocate-A($\cdot$) cannot change the allocations of the previous tasks in $[1, \tth_{l}]$; for all $t\in [\tth_{l^{\prime}}+1, \tth_{l^{\prime}+1}]$ where $l^{\prime}\in [l, K]$, during the execution of the iteration of AllocateRLM($\cdot$) at $t$, we have $t^{\prime}>\tth_{l^{\prime}}$. Hence, the change to the allocations of the previous tasks can only happen in the interval $[\tth_{l^{\prime}}+1, \tth_{l^{\prime}+1}]$.
\end{proof}

From the first and second points of Lemma~\ref{structure-2}, we could conclude that
\begin{proposition}\label{fully-2}
Given a $m\in [1, K]$, $\left[ \tth_{l}+1, d \right]$ is optimally utilized by every task $T_{i} \in \bigcup_{j=1}^{m}{\mathcal{A}_{j}}$ for all $l\in [m, K]$.
\end{proposition}


\section{Applications: Part \Rmnum{2}}
\label{more-app-2}

In this section, we illustrate the applications of the results in Section~\ref{sec.optimal} to (\rmnum{1}) the dynamic programming technique for social welfare maximization, (\rmnum{2}) the machine minimization objective, and (\rmnum{3}) the objective of minimizing the maximum weighted completion time.

\subsection{Dynamic Programming}


For any solution, there must exist a feasible schedule for the tasks selected to be fully allocated by this solution. So, the set of tasks in an optimal solution satisfies the boundary condition by Lemma~\ref{boundary-condition}. Then, to find the optimal solution, we only need address the following problem: if we are given $C$ machines, how can we choose a subset $\mathcal{S}$ of tasks in $\mathcal{T}$ such that (\rmnum{1}) this subset satisfies the boundary condition, and (\rmnum{2}) no other subset of selected tasks achieves a better social welfare? This problem can be solved via dynamic programming (DP). To propose a DP algorithm, we need to identify a dominant condition for the model of this paper \cite{Williamson}. Let $\mathcal{F}\subseteq\mathcal{T}$ and recall that the notation $\lambda_{m}^{C}(\mathcal{F})$ in Section~\ref{sec.optimal-utilization}. Now, we define a $L$-dimensional vector
\begin{center}
$H(\mathcal{F})=(\lambda_{1}^{C}(\mathcal{F})-\lambda_{0}^{C}(\mathcal{F}), \cdots, \lambda_{L}^{C}(\mathcal{F})-\lambda_{L-1}^{C}(\mathcal{F}))$,
\end{center}
where $\lambda_{m}^{C}(\mathcal{F})-\lambda_{m-1}^{C}(\mathcal{F})$, $m\in[L]^{+}$, denotes the optimal resource that $\mathcal{F}$ can utilize on $C$ machines in the segmented timescale $[\tau_{L-m}+1, \tau_{L-m+1}]$ after $\mathcal{F}$ has utilized $\lambda_{m-1}^{C}(\mathcal{F})$ resource in $[\tau_{L-m+1}+1, \tau_{L}]$. Let $v(\mathcal{F})$ denote the total value of the tasks in $\mathcal{F}$ and then we introduce the notion of one pair $(\mathcal{F}, v(\mathcal{F}))$ {\em dominating} another $(\mathcal{F}^{\prime}, v(\mathcal{F}^{\prime}))$ if $H(\mathcal{F})=H(\mathcal{F}^{\prime})$ and $v(\mathcal{F})\geq v(\mathcal{F}^{\prime})$, that is, the solution to our problem indicated by $(\mathcal{F}, v(\mathcal{F}))$ uses the same amount of resources as $(\mathcal{F}^{\prime}, v(\mathcal{F}^{\prime}))$, but obtains at least as much value.

\begin{algorithm}
\SetKwInOut{Begin}{Begin}
\SetKwInOut{Output}{Output}
\BlankLine
\BlankLine

$\mathcal{F}\leftarrow\{T_{1}\}$\;

$A(1)\leftarrow \{ (\emptyset, 0), (\mathcal{F}, v(\mathcal{F})) \}$\;

\For{$i\leftarrow 2$ \KwTo $n$}{

  $A(j)\leftarrow A(i-1)$\;

  \For{each $(\mathcal{F}, v(\mathcal{F}))\in A(i-1)$}{
    \If{$\{T_{i}\}\cup\mathcal{F}$ satisfies the boundary condition}{
      \If{there exist a pair $(\mathcal{F}^{\prime}, v(\mathcal{F}^{\prime}))\in A(i)$ so that (1) $H(\mathcal{F}^{\prime})=H(\mathcal{F}\cup\{T_{i}\})$, and (2) $v(\mathcal{F}^{\prime})\geq v(\mathcal{F}\cup\{T_{i}\})$}{
        Add $(\{T_{i}\}\cup\mathcal{F}, v(\{T_{i}\}\cup\mathcal{F}))$ to $A(i)$\;
        Remove the dominated pair $(\mathcal{F}^{\prime}, v(\mathcal{F}^{\prime}))$ from $A(i)$\;
      }
      \Else{
        Add $(\{T_{i}\}\cup\mathcal{F}, v(\{T_{i}\}\cup\mathcal{F}))$ to $A(i)$\;
      }
    }
  }
}
return $\arg\max_{(\mathcal{F}, v(\mathcal{F}))\in A(n)}\{v(\mathcal{F})\}$\;
\caption{DP($\mathcal{T}$)\label{Recurrence}}
\end{algorithm}

We now give the general DP procedure DP($\mathcal{T}$), also presented as Algorithm~\ref{Recurrence} \cite{Williamson}. Here, we iteratively construct the lists $A(i)$ for all $i\in[n]^{+}$. Each $A(i)$ is a list of pairs $(\mathcal{F}, v(\mathcal{F}))$, in which $\mathcal{F}$ is a subset of $\{T_{1}, T_{2}, \cdots, T_{i}\}$ satisfying the boundary condition and $v(\mathcal{F})$ is the total value of the tasks in $\mathcal{F}$. Each list only maintains all the dominant pairs. Specifically, we start with $A(1)=\{(\emptyset, 0), (\{T_{1}\}, v_{1})\}$. For each $i=2,\cdots,n$, we first set $A(i)\leftarrow A(i-1)$, and for each $(\mathcal{F}, v(\mathcal{F}))\in A(i-1)$, we add $(\mathcal{F}\cup\{T_{i}\}, v(\mathcal{F}\cup\{T_{i}\}))$ to the list $A(i)$ if $\mathcal{F}\cup\{T_{i}\}$ satisfies the boundary condition. We finally remove from $A(i)$ all the dominated pairs.
DP($\mathcal{T}$) will select a subset $\mathcal{S}$ of $\mathcal{T}$ from all pairs $(\mathcal{F}, v(\mathcal{F}))\in A(n)$ so that $v(\mathcal{F})$ is maximum.


\begin{proposition}\label{selection}
DP($\mathcal{T}$) outputs a subset $\mathcal{S}$ of $\mathcal{T}=\{T_{1},$ $\cdots, T_{n}\}$ such that $v(\mathcal{S})$ is the maximum value subject to the condition that $\mathcal{S}$ satisfies the boundary condition; the time complexity of DP($\mathcal{T}$) is $\mathcal{O}(nd^{L}C^{L})$.
\end{proposition}
\begin{proof}
The proof is similar to the one in the knapsack problem \cite{Williamson}. By induction, we need to prove that $A(i)$ contains all the non-dominated pairs corresponding to feasible sets $\mathcal{F}\in\{T_{1}, \cdots, T_{i}\}$. When $i=1$, the proposition holds obviously. Now suppose it hold for $A(i-1)$. Let $\mathcal{F}^{\prime}\subseteq \{T_{1}, \cdots, T_{i}\}$ and $H(\mathcal{F}^{\prime})$ satisfies the boundary condition. We claim that there is some pair $(\mathcal{F}, v(\mathcal{F}))\in A(i)$ such that $H(\mathcal{F})=H(\mathcal{F}^{\prime})$ and $v(\mathcal{F})\geq v(\mathcal{F}^{\prime})$. First, suppose that $T_{i}\notin\mathcal{F}^{\prime}$. Then, the claim follows by the induction hypothesis and by the fact that we initially set $A(i)$ to $A(i-1)$ and removed dominated pairs. Now suppose that $T_{i}\in\mathcal{F}^{\prime}$ and let $\mathcal{F}_{1}^{\prime}=\mathcal{F}^{\prime}-\{T_{i}\}$.  By the induction hypothesis there is some $(\mathcal{F}_{1}, v(\mathcal{F}_{1}))\in A(i-1)$ that dominates $(\mathcal{F}_{1}^{\prime}, v(\mathcal{F}_{1}^{\prime}))$. Then, the algorithm will add the pair $(\mathcal{F}_{1}\cup\{T_{i}\}, v(\mathcal{F}_{1}\cup\{T_{i}\}))$ to $A(i)$. Thus, there will be some pair $(\mathcal{F}, v(\mathcal{F}))\in A(i)$ that dominates $(\mathcal{F}^{\prime}, v(\mathcal{F}^{\prime}))$. Since the size of the space of $H(\mathcal{F})$ is no more than $(C\cdot T)^{L}$, the time complexity of DP($\mathcal{T}$) is $nd^{L}C^{L}$. 
\end{proof}

\begin{proposition}\label{proposition-DP}
Given the subset $\mathcal{S}$ output by DP($\mathcal{T}$), LDF($\mathcal{S}$) gives an optimal solution to the welfare maximization problem with a time complexity $\mathcal{O}(\max\{nd^{L}C^{L}, n^{2}\})$.
\end{proposition}
\begin{proof}
It follows from Propositions~\ref{selection}~and~\ref{property-task}.
\end{proof}

\noindent\textit{Remark.} As in the knapsack problem \cite{Williamson}, to construct the algorithm DP($\mathcal{T}$), the pairs of the possible state of resource utilization and the corresponding best social welfare have to be maintained and a $L$-dimensional vector has to be defined to indicate the resource utilization state. This seems to imply that we cannot make the time complexity of a DP algorithm polynomial in $L$.

\subsection{Machine Minimization}


Given a set of tasks $\mathcal{T}$, the minimal number of machines needed to produce a feasible schedule of $\mathcal{T}$ is exactly the minimum $C^{*}$ such that the boundary condition is satisfied, by Theorem~\ref{main-theorem}, where the feasible schedule could be produced with a time complexity $\mathcal{O}(n^{2})$. An upper bound of the minimum $C^{*}$ is $k\cdot n$ and this minimum $C^{*}$ can be obtained through a binary search procedure with a time complexity of $\log{(k\cdot n)}=\mathcal{O}(\log{n})$; the corresponding algorithm is presented as Algorithm~\ref{minimization}.

\begin{lemma}\label{complexity-1}
In each iteration of the binary search procedure, the time complexity of determining the satisfiability of boundary condition (line 4 of Algorithm~\ref{minimization}) is $\mathcal{O}(L\cdot n)$ where $L\leq n$.
\end{lemma}
\begin{proof}
See the Appendix for the proof. 
\end{proof}

With Lemma~\ref{complexity-1}, the loop of Algorithm~\ref{minimization} has a complexity $\mathcal{O}(L\cdot n\cdot \log{n})$. Based on the above discussion, we conclude that

\begin{proposition}\label{machine-minimization}
Algorithm~\ref{minimization} produces an exact algorithm for the machine minimization problem with a time complexity of $\mathcal{O}(n^{2}, L\cdot n\cdot \log{n})$.
\end{proposition}

\begin{algorithm}
\SetKwInOut{Begin}{Begin}
\SetKwInOut{Output}{Output}
\BlankLine
\BlankLine

$L\leftarrow 1$,  $U\leftarrow k\cdot n$\tcp*{\footnotesize{$L$ and $U$ are respectively the lower and upper bounds of the minimum number of needed machines}}

$mid\leftarrow \frac{L+U}{2}$\;

\While{$U-L\leq 1$}{

   \If{the boundary condition is satisfied with $C=C^{*}$}{
   
       $U \leftarrow mid$\;
   }
   \Else{
   
       $L \leftarrow mid$\;
    
   } 
   
   $mid\leftarrow \frac{L+U}{2}$\;   

}

$C^{*}\leftarrow U$\tcp*{\footnotesize{the optimal number of machines}}

call the algorithm LDF($\mathcal{T}$) to produce a schedule of $\mathcal{T}$ on $C^{*}$ machines\;

\caption{Machine Minimization)\label{minimization}}
\end{algorithm}

\subsection{Minimizing Maximum Weighted Completion Time}

Under the task model of this paper and for the objective of minimizing the maximum weighted completion time of tasks, a direction application of LDF($\mathcal{S}$) improves the algorithm in \cite{Nagarajan} by a factor 2. In \cite{Nagarajan}, with a polynomial time complexity, Nagarajan {\em et al.} find a completion time $d_{i}$ for each task $T_{i}$ that is $1+\epsilon$ times the optimal in terms of the objective here; then they propose a scheduling algorithm where each task can be completed by the time at most 2 times $d_{i}$. As a result, an $(2+2\epsilon)$-approximation algorithm is obtained. Instead, by using the optimal scheduling algorithm LDF($\mathcal{S}$), we have that
\begin{proposition}
There is a ($1+\epsilon$)-approximation algorithm for scheduling independent malleable tasks under the objective of minimizing the maximum weighted completion time of tasks.
\end{proposition}



\section{Conclusion}
\label{sec.conclusion}

In this paper, we study the problem of scheduling $n$ deadline-sensitive malleable batch tasks on $C$ identical machines. Our core result is a new theory to give the first optimal scheduling algorithm so that $C$ machines can be optimally utilized by a set of batch tasks. We further derive four algorithmic results in obvious or non-obvious ways: (\rmnum{1}) the best possible greedy algorithm for social welfare maximization with a polynomial time complexity of $\mathcal{O}(n^{2})$ that achieves an approximation ratio of $\frac{s-1}{s}$, (\rmnum{2}) the first dynamic programming algorithm for social welfare maximization with a polynomial time complexity of $\mathcal{O}(\max\{nd^{L}C^{L}, n^{2}\})$, (\rmnum{3}) the first exact algorithm for machine minimization with a polynomial time complexity of $\mathcal{O}(n^{2}, L\cdot n\cdot \log{n})$, and (\rmnum{4}) an improved polynomial time approximation algorithm for the objective of minimizing the maximum weighted completion time of tasks, reducing the previous approximation ratio by a factor 2. Here, $L$ and $d$ are the number of deadlines and the maximum deadline of tasks.

\appendix

\begin{proof}[Proof of Lemma~\ref{exit-condition-1}]
Before executing Fully-Utilize($i$), the resource allocation to the previously allocated tasks $\mathcal{S}^{\prime}$ satisfies Property~\ref{proper-1}. Its execution does not change the previous allocation to $\mathcal{S}^{\prime}$. Let $\mathcal{S}^{\prime\prime}=\mathcal{S}^{\prime}\cup\{T_{i}\}$. Since $d_{i}=\tau_{m}$, the workload of $T_{i}$ can only be processed in $[1, \tau_{m}]$; the maximum workload of $\mathcal{S}^{\prime\prime}$ that could be processed in $[\tau_{m}+1, \tau_{L}]$ still equals its counterpart when $\mathcal{S}^{\prime}$ is considered, i.e., $\lambda_{L-m}^{C}(\mathcal{S}^{\prime\prime})=\lambda_{L-m}^{C}(\mathcal{S}^{\prime})$. Upon completion of Fully-Utilize($i$), if the total allocation to $\mathcal{S}^{\prime\prime}$ in $[1, \tau_{m}]$ is $C\cdot\tau_{m}$, we could conclude that $T_{i}$ is the last task of $\mathcal{S}$ being considered and all tasks in $\mathcal{S}$ have been fully allocated; otherwise, $\mathcal{S}^{\prime\prime}\subsetneq\mathcal{S}$, which contradicts the fact that $\mathcal{S}$ and its subset satisfy the boundary condition, which implies that after the maximum workload of $\mathcal{S}^{\prime\prime}$ has been processed in $[\tau_{m}+1, \tau_{L}]$, the remaining workload $\mu_{m}^{C}(\mathcal{S}^{\prime\prime})\leq \cdot C\cdot\tau_{m}$. Hence, we conclude that
\end{proof}

\begin{proof}[Proof of Lemma~\ref{lemma-fully-utilize}]
During the execution of Fully-Utilize($i$), upon completion of the allocation to $T_{i}$ at $t\in [1, t_{1}]$, if $T_{i}$ has not been fully allocated yet, it is allocated $k_{i}$ machines at this slot.
The allocations to $T_{i}$ at slots $t_{1}, \cdots, 1$ are non-increasing, i.e.,
\begin{center}
$y_{i}(1)\leq y_{i}(2)\leq \cdots\leq y_{i}(t_{1})$.
\end{center}
The reason for this is as follows: Fully-Utilize($i$) allocates machines to $T_{i}$ from $d_{i}$ towards earlier slots and, after the allocation at every slot $t\in [1, t_{1}]$, $y_{i}(t)=\min\{k_{i}, D_{i}-\sum_{\overline{t}=t+1}^{d_{i}}{y_{i}(\overline{t})}\}$ whose value is non-increasing with $t$. With Property~\ref{proper-2}, before executing Fully-Utilize($i$), the numbers of idle machines have a stepped shape, i.e., $\overline{W}(1)\geq \cdots\geq \overline{W}(t_{0})$. The execution of Fully-Utilize($i$) does not change the previous allocation to $\mathcal{S}^{\prime}$ and upon its completion the number of available machines $\overline{W}(t)$ at every slot $t\in [1, \tau_{m}]$ will be no larger than its counterpart before executing Fully-Utilize($i$); we thus have $t_{0}\geq t_{1}$. Upon completion of Fully-Utilize($i$), deducting the machines allocated to $T_{i}$, the numbers of idle machines still have a stepped shape in $[1, t_{1}]$. Hence, the lemma holds.
\end{proof}

\begin{proof}[Proof of Lemma~\ref{lemma-existence}]
Recall that $W(t)$ is the sum of the allocations $y_{j}(t)$ of all tasks $T_{j}\in \mathcal{S}$ at $t$ and $\overline{W}(t) + W(t)=C$. Initially, we have the inequality that $W(t)-y_{i}(t)>W(t^{\prime})-y_{i}(t^{\prime})$ due to the conditions (\rmnum{1})-(\rmnum{3}) of Lemma~\ref{lemma-existence}, and, there exists a $T_{i^{\prime}}$ such that $y_{i^{\prime}}(t^{\prime})<y_{i^{\prime}}(t)$; otherwise, that inequality would not hold. In the subsequent iteration of Routine($\cdot$), $\overline{W}(t)$ becomes $>0$ since partial allocation of $T_{i^{\prime}}$ is transferred from $t$ to $t^{\prime}$; however, it still holds that $\overline{W}(t)<\Delta\leq k_{i}-y_{i}(t)$. So, we have
\begin{center}
$W(t)-y_{i}(t) = C-\overline{W}(t)-y_{i}(t) > W(t^{\prime})-k_{i} = W(t^{\prime})-y_{i}(t^{\prime})$
\end{center}
and such $T_{i^{\prime}}$ can still be found like the initial case.
\end{proof}

\begin{proof}[Proof of Lemma~\ref{lemma-allocate}]
 If $T_{i}$ has been allocated $D_{i}$ resource just upon completion of Fully-Utilize($\cdot$), Fully-Allocate($i$) does nothing upon its completion and we have $t_{2}=t_{1}$ and the lemma holds. Otherwise, within $[1, \tau_{m}]$, by Lemma~\ref{lemma-fully-utilize}, only the time slots $\overline{t}$ in $[1, t_{1}]$ have available machines, i.e., $\overline{W}(\overline{t})>0$, and, at these time slots, $y_{i}(\overline{t})=k_{i}$; for all $\overline{t}\in[t_{1}+1, d_{i}]$, $\overline{W}(\overline{t})=0$. So, only for each $t$ in $[t_{1}+1, d_{i}]$ and from $d_{i}$ towards earlier time slots, Fully-Allocate($i$) will reduce the allocations of the previous tasks of $\mathcal{S}^{\prime}$ at $t$ and transfer them to the latest time slot $t^{\prime}$ in $[1, t_{1}]$ with $\overline{W}(t^{\prime})>0$ (see the step 2 of Fully-Allocate($i$)); then, all the available machines at $t$ will be re-allocated to $T_{i}$ and $\overline{W}(t)$ is still zero again (see the step 3 of Fully-Allocate($i$)), and, the number of available machines at $t^{\prime}$ will be decreased to zero one by one from $t_{1}$ toward earlier time slots. Due to Lemma~\ref{lemma-fully-utilize}, the lemma holds.
\end{proof}

\begin{proof}[Proof of Lemma~\ref{lemma-complexity}]
The time complexity of Allocate-B($i$) depends on Fully-Allocate($i$) or AllocateRLM($\cdot$). In the worst case, Fully-Allocate($i$) and AllocateRLM($\cdot$) have the same time complexity from the execution of Routine($\cdot$) at every time slot $t\in[1, d_{i}]$. In AllocateRLM($\cdot$) for every task $T_{i}\in\mathcal{T}$, each loop iteration at $t\in[1, d_{i}]$ needs to seek the time slot $t^{\prime}$ and the task $T_{i^{\prime}}$ at most $D_{i}$ times. The time complexities of respectively seeking $t^{\prime}$ and $T_{i^{\prime}}$ are $\mathcal{O}(d)$ and $\mathcal{O}(n)$; the maximum of these two complexities is $\max\{d, n\}$. Since $d_{i}\leq d$ and $D_{i}\leq D$, we have that both the time complexity of Allocate-B($i$) is $\mathcal{O}(dD\max\{d, n\})$. Since we assume that $d$ and $k$ are finitely bounded where $D\leq d\cdot k$, we conclude that $\mathcal{O}(dD\max\{d, n\})=\mathcal{O}(n)$.
\end{proof}

\begin{proof}[Proof of Lemma~\ref{lemma-upper-bound-1}]
 Let us consider an optimal allocation to $\mathcal{A}_1, \mathcal{R}_1, \cdots, \mathcal{A}_K, \mathcal{R}_{K}$ for the MSW-\Rmnum{1} problem. If we replace an allocation to a task in $\mathcal{R}_m$ with the same allocation to a task in $\mathcal{R}_{m}^{\prime}$ and do not change the allocation to $\mathcal{A}_m$, this generates a feasible schedule for the MSW-\Rmnum{2} problem, which yields at least the same social welfare since the marginal value of the task in $\mathcal{R}_{m}^{\prime}$ is no smaller than the ones of the tasks in $\mathcal{R}_m$; hence, Lemma~\ref{lemma-upper-bound-1} holds.
\end{proof}


\begin{proof}[Proof of Lemma~\ref{lemma-upper-bound-2}]
We will show in an optimal schedule of the MSW-\Rmnum{2} problem that (\rmnum{1}) only the tasks of $\mathcal{R}_{m}^{\prime}$, $\mathcal{A}_{1}, \cdots, \mathcal{A}_{K}$ will be executed in $[t_{m-1}^{th}+1, t_{m}^{th}]$, and (\rmnum{2}) the upper bound of the maximum workload of $\mathcal{R}_{m}^{\prime}$ that could be processed in $[t_{m-1}^{th}+1, t_{m}^{th}]$ is $(1-r)\cdot(t_{m}^{th}-t_{m-1}^{th})\cdot C$. As a result, the total value generated by executing all tasks of $\mathcal{A}_{1}, \cdots, \mathcal{A}_{K}$ and $(1-r)\cdot(t_{m}^{th}-t_{m-1}^{th})\cdot C$ workload of each $\mathcal{R}_{m}^{\prime}$ ($m\in [K]^{+}$) is an upper bound of the optimal social welfare for the MSW-\Rmnum{2} problem.

We prove the first point by contradiction. Given a $m\in [K]^{+}$, if $m\geq 2$, all tasks of $\mathcal{R}_{1}^{\prime}, \cdots, \mathcal{R}_{m-1}^{\prime}$ could not be processed in $[t_{m-1}^{th}+1, t_{m}^{th}]$ due to the deadline constraint. If $m\leq K-1$, the marginal value of the task in $\mathcal{R}_{m}^{\prime}$ is no smaller than the ones of $\mathcal{R}_{m+1}^{\prime}, \cdots, \mathcal{R}_{K}^{\prime}$; instead of processing $\mathcal{R}_{m+1}^{\prime}, \cdots, \mathcal{R}_{K}^{\prime}$ in $[t_{m-1}^{th}+1, t_{m}^{th}]$, processing $\mathcal{R}_{m}^{\prime}$ could generate at least the same value or even a higher value. Hence, the first point holds.

We prove the second point also by contradiction. If there exists a $m^{\prime}\in [1, K]$ such that more than $(1-r)\cdot(t_{m^{\prime}}^{th}-t_{m^{\prime}-1}^{th})\cdot C$ workload of $\mathcal{R}_{m^{\prime}}^{\prime}$ is processed in $[t_{m^{\prime}-1}^{th}+1, t_{m^{\prime}}^{th}]$, let $m$ denote the minimum such $m^{\prime}$. In the case where $m=1$, due to Features~\ref{fully} and~\ref{utilization1}, after the maximum workload of the tasks of $\mathcal{A}_{1}$ has been processed in $[t_{1}^{th}+1, t_{K}^{th}]$, the minimum remaining workload that could be processed in $[1, t_{1}^{th}]$ is at least $r\cdot t_{1}^{th}\cdot C$. If more than $(1-r)\cdot t_{1}^{th}\cdot C$ workload of $\mathcal{R}_{1}^{\prime}$ is processed in $[1, t_{1}^{th}]$, this means that the total amount of workload of $\mathcal{A}_{1}$ processed in $[1, t_{1}^{th}]$ is smaller than $r\cdot t_{1}^{th} \cdot C$; in this case, we could always remove the allocation to $\mathcal{R}_{1}^{\prime}$ and add more allocation to $\mathcal{A}_{1}$ to increase the total value. As a result, the second point holds when $m=1$. In the other case where $m\geq 2$, since we are seeking for an upper bound, we could assume that for all $l\in [m-1]^{+}$, $(1-r)\cdot(t_{l}^{th}-t_{l-1}^{th})\cdot C$ workload of $\mathcal{R}_{l}^{\prime}$ is processed in $[t_{l-1}^{th}+1, t_{l}^{th}]$. Again due to Features~\ref{fully} and~\ref{utilization1}, similar to the case where $m=1$, the minimum available workload of $\sum_{T_{i}\in \cup_{l=1}^{m}{\mathcal{A}_{l}}}$ that could be processed in $[1, t_{m}^{th}]$ is at least $r\cdot t_{m}^{th} \cdot C$. In this case, we could still remove the allocation to $\mathcal{R}_{m}^{\prime}$ and add more allocation to $\sum_{T_{i}\in \cup_{l=1}^{m}{\mathcal{A}_{l}}}$ to increase the total value, with the second point holding when $m\geq 2$.
\end{proof}

\begin{proof}[Proof of Lemma~\ref{lemma-interval-value}]
It suffices to prove that, the total allocation to $\cup_{l=1}^{m}{\mathcal{A}_{l}}$ in $[1, t_{m}^{th}]$ could be divided into $m$ parts such that, for all $l\in [1, m]$, (\rmnum{1}) the $l$-th part has a size $r\cdot (t_{l}^{th}-t_{l-1}^{th})\cdot C$, and (\rmnum{2}) the allocation of the $l$-th part is associated with marginal values no smaller than $\overline{v}_{l}^{\prime}$. Then, the total value generated by executing the $l$-th part is no smaller than $\frac{1-r}{r}$ times the total value generated by the allocation to $\mathcal{R}_{l}^{\prime}$ in $[t_{l-1}^{th}+1, t_{l}^{th}]$. As a result, the value generated by the total allocation to $\cup_{l=1}^{m}{\mathcal{A}_{l}}$ in $[1, t_{m}^{th}]$ is no smaller than $\frac{1-r}{r}$ times the value generated by the allocation to $T_{1}^{\prime}, \cdots, T_{m}^{\prime}$.

Due to Feature~\ref{utilization1}, the allocation to $\mathcal{A}_{1}$ achieves a utilization $r$ in $[1, t_{1}^{th}]$ and we could use a part of this allocation as the first part whose size is $r\cdot t_{1}^{th}\cdot C$. Next, the allocation to $\mathcal{A}_{1}\cup\mathcal{A}_{2}$ achieves a utilization $r$ in $[1, t_{2}^{th}]$; we could deduct the allocation used for the first part and get a part of the remaining allocation to $\mathcal{A}_{1}\cup\mathcal{A}_{2}$ as the second part, whose size is $r\cdot (t_ {2}^{th}-t_ {1}^{th})\cdot C$. Similarly, we could get the 3rd, $\cdots$, $m$-th parts that satisfy the first point mentioned at the beginning of this proof. Since the marginal value of the task of $\mathcal{R}_{l}^{\prime}$ is no larger than the ones of the tasks in $\cup_{l^{\prime}=1}^{l}{\mathcal{A}_{l^{\prime}}}$ for all $1\leq l\leq m$, the second point mentioned above also holds.
\end{proof}

\begin{proof}[Proof of Proposition~\ref{proposition-utilization-2}]
We first show that the resource utilization of $\mathcal{A}_{1}\cup\cdots\cup\mathcal{A}_{m}$ in $[1, \tau_{m}]$ is $r$ upon completion of the $m$-th phase of GreedyRLM; then, we consider a task $T_{i} \in \cup_{l=1}^{m}\mathcal{R}_{l}$ such that $d_{i} = c_{m}$. Since $T_{i}$ is not accepted when being considered, it means that $ \sum_{t\leq d_{i}}{\min\{k_{i}, \overline{W}(t)\}}<D_{i}$ at that time and there are at most $len_{i}-1=\lceil \frac{d_{i}}{s_{i}}\rceil-1$ time slots $t$ with $\overline{W}(t)\geq k_{i}$ in $[1, c_{m}]$. Then, we assume that the number of the time slots $t$ with $\overline{W}(t)\geq k_{i}$ is $\mu$. Since $T_{i}$ isn't fully allocated, we have the current resource utilization of $\mathcal{A}_{1}\cup\cdots\cup\mathcal{A}_{m^{\prime}}$ in $[1, c_{m}]$ is at least
\begin{align*}
&\frac{C  d_{i}-\mu C - (D_{i}-\mu k_{i})}{C\cdot d_{i}} \geq  \frac{Cd_{i}-D_{i} - (len_{i}-1)(C-k_{i})}{C\cdot d_{i}}\\
\geq  &  \frac{C(d_{i}-len_{i}) + (C-k_{i}) + (len_{i}k_{i}-D_{i})}{C\cdot d_{i}} \geq\frac{s-1}{s}\geq r.
\end{align*}

We assume that $T_{i}\in\mathcal{R}_{m^{\prime}}$ for some $m^{\prime}\in[m]^{+}$. Now, we show that, after $T_{i}$ is considered and rejected, the subsequent resource allocation by Allocate-A($j$) to each task $T_{j}$ of $\cup_{l=m^{\prime}+1}^{L}{\mathcal{A}_{l}}$ doesn't change the utilization in $[1, \tau_{m}]$. Fully-Utilize($j$) does not change the allocation to the previous accepted tasks; the operations of changing the allocation to other tasks in AllocateRLM($j$, 0, $t_{m}^{th}+2$) happen in its call to Routine($\Delta$, 0, 1, $t$) where we have $c_{m^{\prime}}\leq t_{m^{\prime}}^{th}\leq t_{l}^{th}$ for all $m^{\prime}+1\leq l\leq L$. Due to the function of lines 6-8 of Routine($\Delta$, 0, 1, $t$), in the $l$-th phase of GreedyRLM, the call to any Allocate-A($j$) will never change the current allocation of $\mathcal{A}_{1}\cup\cdots\cup\mathcal{A}_{m^{\prime}}$ in $[1, c_{m}]$. Hence, if $t^{th}_m = c_m$, upon completion of GreedyRLM, the resource utilization of $\mathcal{A}_{1}\cup\cdots\cup\mathcal{A}_{m}$ where $m^{\prime}\leq m$; if $t^{th}_{m} > c_{m}$, since each time slot in $[c_m+1, t^{th}_m]$ is fully utilized by the definition of $t^{th}_{m}$, the resource utilization in $[c_{m}+1, t^{th}_m]$ is 1 and the final resource utilization will also be at least $r$.
\end{proof}

\begin{proof}[Proof of Lemma~\ref{complexity-1}]
Recall the process of defining $\mu_{m}^{C}(\mathcal{S})$ where $\mathcal{S}=\mathcal{T}$. In Definition~\ref{Def-1} that defines $\lambda_{m}(\mathcal{T})$, $n$ tasks are considered sequentially for each $m\in [L]^{+}$, leading to a complexity $L\cdot n$. In Definition~\ref{Def-2} that derives $\lambda_{m}^{C}(\mathcal{T})$ from $\lambda_{m}(\mathcal{T})$, $\lambda_{1}^{C}(\mathcal{T})$, $\lambda_{2}^{C}(\mathcal{T})$, $\cdots$, $\lambda_{L}^{C}(\mathcal{T})$ are considered sequentially, leading to a complexity $\mathcal{O}(L)$. Finally, $\mu_{m}^{C}(\mathcal{T})=\sum_{T_{i}\in\mathcal{T}}{D_{i}} - \lambda_{m}^{C}(\mathcal{T})$. Hence, the time complexity of determining the satisfiability of boundary condition depends on Definition~\ref{Def-1} and is $\mathcal{O}(L\cdot n)$.
\end{proof}






\end{document}